\documentclass[12pt]{article}

\usepackage[english]{babel}

\usepackage[round]{natbib}
\usepackage{amsmath,amsthm,amsfonts, amssymb}
\usepackage{multirow}
\usepackage{graphicx}
\usepackage{tikz}
\usepackage{setspace}
\usepackage{bm}
\usepackage{bbm}
\usepackage[colorlinks=true, allcolors=blue]{hyperref}
\usepackage{booktabs} 
\usepackage{siunitx} 
\usepackage{authblk}

\usepackage{geometry}
\newcommand{\nn}{\nonumber}

\def\sn{\sum_{i=1}^n}
\def\ba{\mathbf{a}}
\def\bb{\mathbf{b}}
\def\bt{\mathbf{t}}
\def\br{\mathbf{r}}

\def\EE{\mathbb{E}}
\def\PP{\mathbb{P}}
\def\cB{\mathcal{B}}

\def\cM{\mathcal{M}}

\def\cP{\mathcal{P}}
\def\cS{\mathcal{S}}
\def\cI{\mathcal{I}}

\def\HH{\mathbb{H}}
\def\RR{\mathbb{R}}

\usepackage{relsize}
\newcommand{\T}{{\mathsmaller {\rm T}}}

\def\##1\#{\begin{align}#1\end{align}}
\def\$#1\${\begin{align*}#1\end{align*}}
\newcommand{\argmin}{\mathop{\mathrm{argmin}}}

\numberwithin{equation}{section}

\newtheorem{theorem}{Theorem}[section]
\newtheorem{proposition}{Proposition}[section]
\newtheorem{lemma}{Lemma}[section]
\newtheorem{assumption}{Assumption}
\newtheorem{gmm assumption}{GMM Assumption}
\newtheorem{definition}{Definition}[section]
\newtheorem{remark}{Remark}[section]

\newcommand{\myalpha}{\nu}

\title{Fortified Proximal Causal Inference \\ with Many Invalid Proxies}

\author[1]{Myeonghun Yu\footnote{Corresponding author. E-mail: \href{mailto:audgns@umich.edu}{audgns@umich.edu}.}\hspace{.2cm}}
\author[1]{Xu Shi}
\author[2]{Eric J. Tchetgen Tchetgen}

\affil[1]{Department of Biostatistics, University of Michigan, Ann Arbor, MI, 48109, USA}
\affil[2]{Department of Statistics \& Data Science, the Wharton School, University of Pennsylvania, Philadelphia, PA, 19104, USA}

\begin{document}
\date{}
\maketitle

\begin{abstract}
Causal inference from observational data often relies on the assumption of no unmeasured confounding, an assumption frequently violated in practice due to unobserved or poorly measured covariates.
Proximal causal inference (PCI) offers a promising framework for addressing unmeasured confounding using a pair of outcome and treatment confounding proxies.
However, existing PCI methods typically assume all specified proxies are valid, which may be unrealistic and is untestable without extra assumptions.
In this paper, we develop a semiparametric approach for a many-proxy PCI setting that accommodates potentially invalid treatment confounding proxies. 
We introduce a new class of \emph{fortified} confounding bridge functions and establish nonparametric identification of the population average treatment effect (ATE) under the assumption that at least $\gamma$ out of $K$ candidate treatment confounding proxies are valid, for any $\gamma \leq K$ set by the analyst without requiring knowledge of which proxies are valid.
We establish a local semiparametric efficiency bound and develop a class of multiply robust, locally efficient estimators for the ATE. These estimators are thus simultaneously robust to invalid treatment confounding proxies and model misspecification of nuisance parameters.  
The proposed methods are evaluated through simulation and applied to assess the effect of right heart catheterization in critically ill patients.
\end{abstract}


\doublespacing
\section{Introduction}

The assumption of no unmeasured confounding is a common assumption of causal inference in observational studies, which requires that all confounders are measured accurately and accounted for in the analysis.  However, this assumption is often violated in practice, as some confounders may be unobserved or imperfectly measured. Violations of this assumption can lead to biased estimates and invalid inference about causal effects.

To address this challenge, a variety of methods have been developed over the years, including recent approaches that leverage so-called proxy variables which have drawn growing interest \citep{miao2018identifying, shi2020selective,tchetgen2024anint,cui2024}.
Proximal causal inference (PCI) provides a formal framework that uses proxy variables to identify and estimate causal effects despite the presence of unmeasured confounding. 
In addition to having potentially measured a subset of relevant confounders, it relies on having also measured two types of proxy variables: (1) treatment confounding proxies; and (2) outcome confounding proxies. 
Treatment confounding proxies must be associated with the outcome only through a shared unmeasured confounder with the treatment, while outcome confounding proxies must be associated with the treatment only through a shared unmeasured confounder with the outcome. 
We refer readers to Section~\ref{sec:preliminaries} for formal definitions.
Prominent examples of proxies of types (1) and (2) are so-called negative control exposures and negative control outcomes, respectively. 
Importantly, both types of proxies should share the same unmeasured confounding structure as the primary treatment and outcome of interest.


However, the assumption that proxy variables are valid, in the sense of satisfying either (1) or (2), is not empirically testable without additional assumptions.  Assessments of proxy validity often rely on domain knowledge \citep{lipsitch2010negative,shi2020selective}, which may be subjective and potentially unreliable. Existing data-driven validity tests rely on stringent structural equation models, and are limited to validating one particular type of proxies, known as \emph{disconnected} proxies \citep{kummerfeld2024data}.
In practice, identifying valid outcome confounding proxies may be more straightforward than selecting valid treatment confounding proxies.
For example, pre-treatment outcome measures are widely used as outcome confounding proxies since they cannot be causally affected by the treatment. 
A case in point is \citet{jackson2006evidence}, who used pre-season influenza hospitalizations as an outcome confounding proxy to detect residual confounding bias in evaluating the effect of vaccination among the elderly. More broadly, the widely popular difference-in-differences methodology uses pre-treatment outcome measurements as valid outcome confounding proxies to debias causal effect estimates \citep{sofer2016negative}. 
In contrast, identifying valid treatment confounding proxies may be considerably more challenging in practice. 
The presence of invalid proxies can result in biased estimation and misleading inference, yet existing PCI literature offers limited methodological tools for dealing with this problem.


In this paper,  assuming that valid outcome confounding proxies are available, we develop a semiparametric framework for PCI about the population average treatment effect (ATE) that accommodates possibly invalid treatment confounding proxies. 
After introducing a new class of \emph{fortified} confounding bridge functions \citep{miao2018identifying}, we establish nonparametric identification of the ATE under the assumption that at least $\gamma$ out of $K$ candidate treatment confounding proxies are valid, for $1 \leq \gamma \leq K$ a priori set by the analyst, without necessarily knowing their identity.
As such, our proposed approach does not rely on the stringent assumption common in the proximal literature that all specified proxies are valid, providing identifiability conditions as long as $\gamma$ is a valid lower bound for the actual number of valid treatment confounding proxies. 
Beyond identification, we also consider estimation and inference in the proposed framework, and we characterize a locally semiparametric efficiency bound for the ATE under our proposed model and derive a rich class of multiply robust, locally efficient estimators. 
These estimators remain consistent if one of several working models is correctly specified, and attain the semiparametric efficiency bound in the union model when all working models are correctly specified. Notably, our method avoids an explicit model selection step, that is, it does not attempt to identify which proxies are invalid. However, as the true number of valid proxies may not be available, the analyst can probe the robustness of empirical results by evaluating whether estimated causal effects remain stable across different values of $\gamma$ over a certain range in a form of sensitivity analysis, given that the identified parameter should be constant provided that $\gamma$ is a valid lower bound for the underlying true number of valid proxies.  A similar strategy may be adopted when many outcome confounding proxies are available, to evaluate whether some are invalid, by repeating the analysis and assessing the sensitivity of results across different subsets of outcome confounding proxies. 

The remainder of the paper is organized as follows.
Section~\ref{sec2} reviews existing identification results for the ATE under the assumption that the given outcome and treatment confounding proxies are valid \citep{miao2018identifying, cui2024}, and then introduces our nonparametric identification strategy that allows for some invalid treatment confounding proxies, without requiring knowledge of which ones are valid, ex-ante.
In Section~\ref{sec:3}, we develop semiparametric theory for PCI under this framework and construct corresponding multiply robust, locally efficient estimators.
Section~\ref{sec:simulation} presents simulation studies, and Section~\ref{sec:data.analysis} illustrates a real data application evaluating the effectiveness of right heart catheterization in critically ill patients admitted to the intensive care unit.
Additional discussion and all technical proofs are provided in the Supplementary Material.

\bigskip
\noindent
{\sc Notation}. 
For any random variable $V$, let $L_2(V)$ be the space of square-integrable functions of $V$ with the $L_2$-norm $\|\cdot\|_2$.
For any finite set $A$, we denote its cardinality by $|A|$ and define $\cP_t(A) = \{\myalpha \subseteq A : |\myalpha| = t\}$ and $\cP_{\geq t}(A) = \{\myalpha \subseteq A : |\myalpha| \geq t\}$.
For any finite set of subspaces $\{S_i: i \in \cI\}$ of a vector space $R$, we denote the sum of these subspaces by $\sum_{i \in \cI} S_i$.

\section{Nonparametric identification}
\label{sec2}

\subsection{Preliminaries} \label{sec:preliminaries}

We consider the estimation of the effect of a binary treatment $A \in \{0,1\}$ on a real-valued outcome variable $Y \in \RR$, in the presence of potential confounding by observed covariates $L$ and unobserved variables $U$.
For $a \in \{0,1\}$, let $Y(a)$ denote the potential outcome that would be observed if the treatment were set to $a$ \citep{splawa1990application,rubin1974estimating}.
Our goal is to estimate the ATE $\tau^*$ defined as $\tau^* := \EE\{Y(1) - Y(0)\}$.
For the identification of $\tau^*$, we begin with the following consistency assumption:
\begin{assumption} [Consistency] \label{assump:consistency}
$Y = Y(a)$ almost surely when $A = a$.
\end{assumption}

Suppose that all confounding factors are observed, ensuring the exchangeability assumption $Y(a) \perp\mspace{-9mu}\perp A|L$ for $a = 0,1$.
Then, under Assumption~\ref{assump:consistency} and the positivity assumption that $0 < \PP(A = 1|L) < 1$ almost surely over $L$, the counterfactual mean $\EE\{Y(a)\}$ can be identified by the well-known g-formula \citep{robins1986new}.

Instead of requiring precise measurement of all confounders to satisfy the exchangeability assumption, the PCI literature acknowledges that measured covariates $L$ are often imperfect proxies of the true underlying confounding mechanism \citep{tchetgen2024anint}. The approach assumes that $L$ includes valid treatment and outcome confounding proxies correctly identified by the analyst, ex-ante. Mainly, $L=(X,Z,W)$, where $X$ represents common causes of treatment and outcome, (i.e. measured confounders), treatment confounding proxies $Z$  are a priori known to be unrelated to the outcome except through $U$ conditional on $A$ and $X$, and outcome confounding proxies $W$ which are unrelated to treatment other than through $U$ conditional on $X$. 
This view has generalized Robins' g-formula to the \emph{proximal g-formula}, which identifies the ATE even in the presence of unmeasured confounding, by carefully leveraging available proxies of hidden confounders. We refer the reader to the rich and fast-growing literature of PCI \citep{miao2018identifying, shi2020multiply, ying2023proximal, cui2024, miao2024confounding, tchetgen2024anint}. An important and long-standing problem of PCI is the need to find and validate treatment and outcome confounding proxies \citep{kummerfeld2024data}.

In this paper, we relax the validity requirement of treatment confounding proxies by assuming that while $W$ is a valid outcome confounding proxy, $Z = (Z_1, \ldots, Z_K)$ consists of $K$ candidate treatment confounding proxies for some positive integer $K$, some of which may be invalid.
Thus, we assume that some or all of these $K$ variables may serve as valid treatment confounding proxies, however their identity is not known a priori. 
To characterize the relationships among these variables formally, we first define the notion of valid treatment confounding proxies among $K$ candidates.
Let $[K] = \{1,2,\ldots,K\}$ and denote $Z_\myalpha = (Z_j: j \in \myalpha)$ and $Z_{-\myalpha} = (Z_j : j \notin \myalpha)$ as the subvectors of $Z$ for any index set $\myalpha \subseteq [K]$.

\begin{definition} \label{def:valid.proxy}
A subset $\myalpha^* \subseteq [K]$ is said to represent a set of indices of valid treatment confounding proxies, or equivalently, $Z_{\myalpha^*}$ are said to be valid treatment confounding proxies among $Z$ if $Z_{\myalpha^*}$ satisfy the following assumptions for $a \in \{0,1\}$:
\begin{itemize}
    \item[(i)] (Latent ignorability) $A \perp\mspace{-9mu}\perp Y(a)|U, Z_{-\myalpha^*}, X$.
    \item[(ii)] (Positivity) $0 < \PP(A = a|U, Z_{-\myalpha^*}, X) < 1$ almost surely.
    \item[(iii)] (Proxies) $Z_{\myalpha^*} \perp\mspace{-9mu}\perp Y|U,A,Z_{-\myalpha^*}, X$ and $W \perp\mspace{-9mu}\perp (Z_{\myalpha^*}, A) | U, Z_{-\myalpha^*}, X$.
\end{itemize}
\end{definition}

\begin{figure}
\centering
\begin{tikzpicture}
\node[circle, draw] (u) at (3,3.5) {$U$};
\node (z2) at (3,2) {$Z_2$};
\node (z1) at (0,1) {$Z_1$};
\node (w) at (6,1) {$W$};
\node (a) at (1.5,0) {$A$};
\node (y) at (4.5,0) {$Y$};

\draw [-stealth] (a) -- (y);
\draw [stealth-] (a) -- (u);
\draw [stealth-] (a) -- (z1);
\draw [stealth-] (z1) -- (u);
\draw [stealth-] (w) -- (u);
\draw [stealth-] (y) -- (u);
\draw [-stealth] (w) -- (y);

\draw [-stealth] (z2) -- (z1);
\draw [-stealth] (u) -- (z2);
\draw [-stealth] (z2) -- (a);
\draw [-stealth] (z2) -- (y);
\draw [-stealth] (z2) -- (w);
\end{tikzpicture}
\caption{A causal DAG when $Z_{\myalpha^*}=Z_1$, which is a valid treatment confounding proxy, and $Z_{-\myalpha^*}=Z_2$ which is an invalid treatment confounding proxy, among two candidates $Z=(Z_1, Z_2)$. Here we omitted $X$ and all relationships are implicitly conditional on $X$.}
\label{fig:dag}
\end{figure}
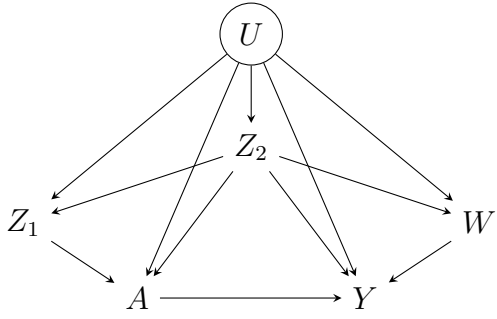

In the above definition, $\myalpha^* \subseteq [K]$ denotes the set of indices of valid treatment confounding proxies, which is unknown to the analyst.
That is,  we classify  $Z=(Z_{\myalpha^*},Z_{-\myalpha^*})$ into two groups: $Z_{\myalpha^*}$, which represents valid treatment confounding proxies, and $Z_{-\myalpha^*}$, which represents the remaining potentially invalid proxies. The latter can essentially be treated and adjusted for as standard measured covariates.  Specifically, the latent ignorability assumption ensures that the unmeasured confounders $U$ and measured covariates $(Z_{-\myalpha^*}, X)$ are sufficient to account for confounding in the relationship between treatment $A$ and potential outcome $Y(a)$. 
The positivity assumption requires that, for all strata defined by $U, Z_{-\myalpha^*}$ and $X$, the probability of receiving a particular treatment $a$ is non-zero. 
The last assumption states that the  $Z_{\myalpha^*}$ and $W$ are valid treatment and outcome confounding proxies, respectively, such that $Z_{\myalpha^*}$ is conditionally independent of $Y$ given $(U,Z_{-\myalpha^*},X)$ and $A$, while $W$ is unrelated with $(Z_{\myalpha^*},A)$, except through the measured and unmeasured confounders $(U,Z_{-\myalpha^*},X)$.
Figure~\ref{fig:dag} illustrates a causal directed acyclic graph (DAG) satisfying the above assumptions when $Z_1$ is a valid treatment confounding proxy and $Z_2$ is an invalid treatment confounding proxy given two candidates $Z = (Z_1, Z_2)$.

Definition~\ref{def:valid.proxy} aligns with the classical definition of proxy variables, where $Z_{\myalpha^*}$ and $W$ are valid treatment and outcome confounding proxies, respectively, and $(Z_{-\myalpha^*}, X)$ represent measured baseline covariates.
When $Z_{\myalpha^*}$ is known a priori, as in the case of an oracle, $\tau^*$ would in principle be nonparametrically identified, provided the following completeness conditions hold \citep{miao2018identifying}.

\begin{assumption} [Completeness] \label{assump:valid.completeness}
Let $g(U)$ be an arbitrary square-integrable function of $U$.
We have the following conditions:
\begin{itemize}
    \item[(i)]  For any $(a,z_{-\myalpha^*}, x)$, $\EE\{g(U)|Z_{\myalpha^*}, A = a, Z_{-\myalpha^*} = z_{-\myalpha^*}, X = x\} = 0$  almost surely if and only if $g(U) = 0$ almost surely.
    \item[(ii)] For any $(a,z_{-\myalpha^*}, x)$, $\EE\{g(U)|W, A = a, Z_{-\myalpha^*} = z_{-\myalpha^*}, X = x\} = 0$ almost surely if and only if $g(U) = 0$ almost surely.
\end{itemize}
\end{assumption}

The completeness condition requires that the treatment confounding proxies $Z_{\myalpha^*}$ and the outcome confounding proxy $W$ are sufficiently associated with $U$, that is $U$-relevant \citep{lipsitch2010negative}, so that any variability in $U$ induces corresponding variability in $Z$ and $W$ conditional on $(A, Z_{-\myalpha^*}, X)$.
For example, if $Z_{\myalpha^*}, W$ and $U$ are categorical variables with respective cardinality $d_W, d_{Z_{\myalpha^*}}$ and $d_U$, the completeness assumption can be viewed as a rank condition, requiring $\min({d_{Z_{\myalpha^*}},d_W}) \geq d_U$ \citep{miao2018identifying, shi2020multiply}.
Completeness is a foundational concept in statistical inference, closely related to the notion of sufficiency \citep{bickel2015mathem}.
Recently, it has been applied to identification problems across various areas, including nonparametric instrumental variable regression \citep{newey2003instrumental}, measurement error models \citep{hu2008instrumental}, long-term causal inference \citep{imbens2024long} and PCI \citep{tchetgen2024anint}.

Now, given the knowledge of $\myalpha^*$ and Assumption~\ref{assump:valid.completeness}, $\tau^*$ is nonparametrically identified by Theorem 1 in \cite{miao2018identifying} and Theorem 3.2 in \cite{cui2024}, which we restate for convenience.

\begin{proposition} \label{prop:oracle.identification}
Let $\myalpha^* \subseteq [K]$ be the set of indices of valid treatment confounding proxies defined in Definition~\ref{def:valid.proxy} and Assumption~\ref{assump:consistency} holds.
\begin{itemize}
    \item[(i)] Suppose that Assumption~\ref{assump:valid.completeness}(i) is satisfied and there exists a function $\tilde h$ satisfying the following integral equation:
    \begin{align}
    \EE(Y|Z, A, X) = \EE\{h(W, A, X, Z_{-\myalpha^*})|Z, A, X\}. \label{def:oracle.outcome.bridge}
    \end{align}
    Then, we have that $\tau^* = \EE\{\tilde h(W,1,X,Z_{-\myalpha^*}) - \tilde h(W,0,X,Z_{-\myalpha^*})\}$.
    \item[(ii)] Suppose that Assumption~\ref{assump:valid.completeness}(ii) is satisfied and there exists a function $\tilde q$ satisfying the following integral equation:
    \#
\EE\{q(Z, A, X)|W, A, X, Z_{-\myalpha^*}\} = \frac{1}{\PP(A|W,X, Z_{-\myalpha^*})}. \label{def:oracle.treatment.bridge}
    \#
    Then, we have that $\tau^* = \EE\{\tilde q(Z,1,X)AY-\tilde q(Z,0,X)(1-A)Y\}$.
\end{itemize}
\end{proposition}
A function $\tilde h$ (not necessarily unique) satisfying \eqref{def:oracle.outcome.bridge} is referred to as an outcome confounding bridge function, which stands-in as a well-calibrated forecast of the potential outcome; likewise, $\tilde q$ (not necessarily unique) satisfying \eqref{def:oracle.treatment.bridge} is referred to as a treatment confounding bridge function, and plays a similar role as inverse propensity score weighting under unconfoundness \citep{cui2024, miao2024confounding} but in the proximal setting.
Proposition~\ref{prop:oracle.identification} demonstrates that in the presence of unmeasured confounders, $\tau^*$ can be identified leveraging \emph{valid} outcome and treatment confounding proxies via outcome bridge or treatment bridge functions strategies.

Proposition~\ref{prop:oracle.identification} depends crucially on the accurate knowledge of $\myalpha^*$.
If $Z_{\myalpha^*}$ do not constitute valid treatment confounding proxies, the identification result breaks down, and applying Proposition~\ref{prop:oracle.identification} with invalid proxies will generally lead to the biased estimation of $\tau^*$.
Therefore, it is important to select valid $\myalpha^*$ for valid estimation and inference about $\tau^*$.
However, verifying the validity of $\myalpha^*$ in the presence of $U$ is challenging. 
For instance, directly testing the exclusion restriction of $Z_{\myalpha^*}$ on $Y$ conditional on both observed and unobserved confounders (Definition~\ref{def:valid.proxy}-(iii)) is often infeasible unless one is willing to impose additional assumptions \citep{kummerfeld2024data}.
Given these difficulties, the validity of proxies is often assessed using domain-specific knowledge, which can be subjective and unreliable.

\subsection{Nonparametric identification with invalid proxies}
\label{sec:2.2}

To address the aforementioned limitations in the current literature, we develop nonparametric identification of $\tau^*$, when the set of indices for valid treatment confounding proxies, $\myalpha^*$, is unknown.
In the absence of knowledge of $\myalpha^*$, we assume that, out of the $K$ candidates, at least some of them are valid, as formalized below.

\begin{assumption} \label{assump:lower.bound}
    The unknown set of indices of valid treatment confounding proxies (defined in Definition~\ref{def:valid.proxy}), $\myalpha^*$, satisfies $|\myalpha^*| \geq \gamma$ for some $\gamma \in [K]$. 
\end{assumption}

Although $\myalpha^*$ is unknown, Assumption~\ref{assump:lower.bound} ensures that at least $\gamma$ out of $K$ candidates are valid, for $\gamma$ a priori known to the analyst. Note that $\gamma$ need not be greater than $K/2$ (or the largest number of invalid proxies), which distinguishes our assumption from the majority (or plurality) rule in the instrumental variable literature \citep{han2008detecting, kang2016instrumental,guo2018confidence}. 
Identification under this setting can be interpreted as identification within a union of viable causal models, each corresponding to a subset $\myalpha^* \subseteq [K]$ satisfying $|\myalpha^*| \geq \gamma$.

A key contribution of the paper is to propose novel classes of bridge functions that identify the ATE under Assumption~\ref{assump:lower.bound}.
Specifically, we present two identification results: one based on a new class of outcome confounding bridge functions, and an alternative approach based on a new class of treatment confounding bridge functions which also provides robustness to invalid treatment confounding proxies.
Building on these results, we develop semiparametric estimators of $\tau^*$ in Section~\ref{sec:3}.

\subsubsection{Identification via fortified outcome confounding bridge functions}
\label{sec:outcome.bridge}

To introduce an outcome confounding bridge function under Assumption~\ref{assump:lower.bound}, we first define the subspace $\HH_\gamma$ of $L_2(Z,A,X)$ as follows:
\#
    \HH_\gamma := \{d \in L_2(Z, A, X): ~ \EE\{d(Z,A,X)|Z_{-\myalpha}, X\} = 0,~\forall \myalpha \in \cP_{\geq \gamma}([K])\}. \label{def:H.gamma}
\#
That is, the subspace $\HH_\gamma$ contains all functions of $(Z, A, X)$ that satisfy $\EE\{d(Z, A, X)|Z_{-\myalpha}, X\} = 0$, for all possible partitions $Z=(Z_{\myalpha},Z_{-\myalpha})$ with $|\myalpha|\geq\gamma$. It is characterized solely by the number $\gamma$ and thus does not require knowledge of $\myalpha^*$.
By the Law of Iterative Expectations, $\HH_\gamma$ is equivalently characterized as
\$
\HH_\gamma := \{d \in L_2(Z, A, X): ~ \EE\{d(Z,A,X)|Z_{-\myalpha}, X\} = 0,~\forall \myalpha \in \cP_{\gamma}([K])\}.
\$

\begin{remark}\label{remark:example}
As an illustrative example of $\HH_\gamma$, consider a simple setting with no baseline covariates $X$. Suppose there are $K = 2$ candidate proxies $Z=(Z_1,Z_2)$, and at least $\gamma = 1$ of them are valid. 
That is, the set of possible valid proxies is $Z_1$, $Z_2$, and $(Z_1,Z_2)$.
Moreover, assume $Z_1, Z_2 \in \{-1, 0, 1\}$ for ease of presentation.
Then, it follows that $\HH_1$ is the linear span of 
\$
\cB_1 = & \{\bar A, \bar A \bar Z_1, \bar A \bar Z_2, \bar A \bar Z_1 \bar Z_2, \bar A \bar{Z_1^2}, \bar A \bar{Z_2^2}, \bar A \bar{Z_1^2}\bar Z_2, \bar A \bar Z_1\bar{Z_2^2}, \bar A \bar{Z_1^2}\bar{Z_2^2}, \\
& \bar Z_1 \bar Z_2, \bar{Z_1^2}\bar Z_2, \bar Z_1 \bar{Z_2^2}, \bar{Z_1^2}\bar{Z_2^2} \},
\$ 
where $\bar A = \omega(A,Z_1, Z_2)\{A - \EE(A)\}$, $\bar Z_i = \omega(A,Z_1, Z_2)\{Z_i - \EE(Z_i)\}$ and $\bar{Z_i^2} = \omega(A,Z_1, Z_2)\{Z_i^2 - \EE(Z_i^2)\}$, for $i = 1,2$.
Here, the weight $\omega(a, z_1, z_2)$ is defined as $
\omega(a,z_1,z_2) = \PP(A = a)\PP(Z_1 = z_1)\PP(Z_2 = z_2)/\PP(A = a, Z_1 = z_1, Z_2 = z_2)$.
\end{remark}


We now formally introduce the outcome confounding bridge function in the presence of potential invalid treatment confounding proxies as follows.
\begin{assumption}
\label{assump:h}
    There exists a solution $h^*$ to the following equation:
\#
\EE[d(Z,A,X)\{Y - h^*(W,A,X)\}] = 0, ~~~~\forall d \in \HH_\gamma. \label{outcome.bridge.moment}
\#
\end{assumption}

We refer to $h^*$ satisfying \eqref{outcome.bridge.moment} as a \textit{fortified outcome confounding bridge function}, distinguishing it from a conventional outcome confounding bridge function as detailed in Proposition \ref{prop:oracle.identification}.
When $\gamma = K$, that is, all candidates are valid treatment confounding proxies, Assumption~\ref{assump:h} coincides with the existence of a conventional outcome confounding bridge function in the PCI literature \cite{miao2018identifying}, as shown in Appendix A.1 of the supplementary material.
As such, fortified outcome confounding bridge functions generalize their conventional counterpart.
For formal technical conditions for the existence of a solution to \eqref{outcome.bridge.moment}, we refer the readers to Appendix A.3.

\begin{remark}\label{remark:example_model}
Following Remark~\ref{remark:example}, suppose $Z_1$ is a valid proxy, while $Z_2$ is invalid, which corresponds to the DAG in Figure~\ref{fig:dag}. 
To illustrate the fortified outcome confounding bridge function satisfying \eqref{outcome.bridge.moment}, consider the following true data-generating mechanism defined by linear structural models:
\$
E[Y|A,W, Z,U] & = A+2Z_2+W + U, \nn \\
E[W|A,Z,U] & = Z_2+U. 
\$
The validity of $Z_1$ is reflected in its absence from both right-hand sides, while $Z_2$ appears in both equations, indicating its invalidity.
One can show that $E[Y-A-2W|A,Z]=Z_2$. 
Thus, by definition of $\HH_\gamma$, we  have $\EE[d(Z,A)\{Y - h^*(W,A)\}] = 0$ for any $d\in \HH_\gamma$, where $h^*(W,A)=A+2W$. 
\end{remark}

Similar to Proposition \ref{prop:oracle.identification}, we introduce the following completeness conditions for nonparametric identification of $\tau^*$ using fortified outcome confounding bridge functions.

\begin{assumption} [Fortified Completeness] \label{assump:outcome.bridge.completeness}
    For any $g \in L_2(U, A, Z_{-\myalpha^*}, X)$, 
    \$
    \EE\{d(Z, A, X)g(U, A, Z_{-\myalpha^*}, X)\} = 0 ~~~\mbox{ for all }~ d \in \HH_\gamma + L_2(Z_{-\myalpha^*}, X)
    \$ 
    if and only if $g = 0$ almost surely.
\end{assumption}
Assumption~\ref{assump:valid.completeness}(i), which is necessary for identification when $\myalpha^*$ is known, is equivalent to requiring that for any $g \in L_2(U, A, Z_{-\myalpha^*}, X)$, 
\$
\EE\{d(Z,A,X)g(U, A, Z_{-\myalpha^*}, X)\} = 0 ~\mbox{ for all } d \in L_2(Z,A,X)
\$ 
if and only if $g = 0$ almost surely.
Since $\HH_\gamma + L_2(Z_{-\myalpha^*}, X) \subseteq L_2(Z,A,X)$, Assumption~\ref{assump:outcome.bridge.completeness} imposes a stronger requirement than Assumption~\ref{assump:valid.completeness}(i), explicitly illustrating the price for identifying $\tau^*$ without knowing $\myalpha^*$.
Heuristically, Assumption~\ref{assump:valid.completeness}(i) requires that $L_2(Z, A, X)$ is sufficiently informative about $L_2(U, A, Z_{-\myalpha}, X)$, while Assumption~\ref{assump:outcome.bridge.completeness} demands the same degree of informativeness for the smaller functional space $\HH_\gamma + L_2(Z_{-\myalpha^*},X)$.
Moreover, for $1 \leq \gamma \leq \gamma^\prime \leq |\myalpha^*|$, we have $\HH_\gamma \subseteq \HH_{\gamma^\prime}$, so the fortified completeness assumption becomes weaker as the lower bound of the number of valid proxies increases. 
When $\gamma = K$, leading to $\HH_K + L_2(X) = L_2(Z,A,X)$, the assumption reduces to the standard completeness assumption in the proximal literature \citep{miao2018identifying,cui2024}.
Thus, Assumption~\ref{assump:outcome.bridge.completeness} highlights explicitly, a trade-off between the lower bound of the number of valid proxies and the strength of the completeness condition required for identification.

\begin{remark} \label{remark:example2}
Under the example of Remark~\ref{remark:example} and assuming $U$ is binary,  $L_2(U, A, Z_2)$ is the linear span of 
\$
\cB_2 = \{1, A, AU, AZ_2, AZ_2^2, AUZ_2, AUZ_2^2, U, Z_2, Z_2^2, UZ_2, UZ_2^2\},
\$ and Assumption~\ref{assump:outcome.bridge.completeness} is equivalent to requiring that the matrix $M \in \RR^{12 \times 15}$ is full rank, where the $(i,j)$-th entry of $M$ is given by the expectation of the product of $i$-th element of $\cB_2$ and $j$-th element of $\cB_1 \cup \{1, Z_2\}$.
\end{remark}

We now establish the nonparametric identification of $\tau^*$ using fortified outcome confounding bridge functions.

\begin{theorem} \label{thm:outcome.bridge.identification}
Suppose that Assumptions~\ref{assump:lower.bound}, \ref{assump:h} and \ref{assump:outcome.bridge.completeness} hold.
Then, any solution $h^*$ to \eqref{outcome.bridge.moment} satisfies
\#
\EE\{Y - h^*(W,A,X)|U,A,Z_{-\myalpha^*},X\} = \EE\{Y - h^*(W,A,X)|Z_{-\myalpha^*},X\}, \label{outcome.bridge.moment.U}
\#
almost surely.
Furthermore, provided that Assumption~\ref{assump:consistency} also holds, the ATE is nonparametrically identified by
\#
\tau^* = \EE\{h^*(W,1,X) - h^*(W,0,X)\}. \label{outcome.identification}
\#
\end{theorem}

Theorem~\ref{thm:outcome.bridge.identification} establishes our first key identification result, about the causal effect using observed variables without prior knowledge of $\myalpha^*$.
Theorem~\ref{thm:outcome.bridge.identification} provides a multiply robust causal identification result by establishing the nonparametric identification within the union of causal models, each corresponding to a partition of $Z=(Z_{\myalpha},Z_{-\myalpha})$, with $|\myalpha|\geq\gamma$.
That is, $\tau^*$ is identified if any of these causal models is valid given the stated regularity conditions, without requiring explicit knowledge or identification of the valid proxies $\myalpha^*$.
As in other PCI setups \citep{tchetgen2024anint,cui2024}, identification does not require a unique solution to \eqref{outcome.bridge.moment}, and thus, in principle, any solution identifies $\tau^*$.
Theorem~\ref{thm:outcome.bridge.identification} demonstrates that a solution of the conditional moment equation \eqref{outcome.bridge.moment.U} can be found without explicitly knowing $\myalpha^*$ or directly modeling or estimating hidden confounding factors.

By definition of $\HH_\gamma$, we establish the equivalent condition for Assumption~\ref{assump:h} as follows.
Let $\cS_\gamma = \overline{\sum_{\myalpha \in \cP_\gamma([K])} L_2(Z_{-\myalpha}, X)}$, where the closure is taken in $L_2(Z,A,X)$.

\begin{proposition} \label{prop:outcome.bridge.identification2}
$h^*$ satisfies \eqref{outcome.bridge.moment} if and only if there exists $h^* \in L_2(W,A,X)$ and $\ell^* \in \cS_\gamma$ such that $\tilde h(W,A,Z,X) := h^*(W,A,X) + \ell^*(Z,A,X)$ satisfies
\#
\EE\{Y - \tilde h(W,A,Z,X) | Z, A, X\} = 0. \label{outcome.bridge.moment2}
\#
Moreover, under Assumptions~\ref{assump:consistency} and \ref{assump:lower.bound}--\ref{assump:outcome.bridge.completeness}, any such solution $\tilde h = h^* + \ell^*$ satisfies $\ell^* \in L_2(Z_{-\myalpha^*},X)$ and  
\#
\EE\{Y - \tilde h(W,A,Z,X) | U,A,Z_{-\myalpha^*}\} = 0, \label{outcome.bridge.moment.U2}
\#
and $\tau^* = \EE\{\tilde h(W,1,Z,X) - \tilde h(W,0,Z,X)\}$.
\end{proposition}

Proposition~\ref{prop:outcome.bridge.identification2} shows that the existence of $h^*$ satisfying \eqref{outcome.bridge.moment} is equivalent to the existence of $\tilde h \in L_2(W,A,X) + \cS_\gamma$ satisfying \eqref{outcome.bridge.moment2}.
Furthermore, under the regularity conditions in Theorem~\ref{thm:outcome.bridge.identification} and Proposition~\ref{prop:outcome.bridge.identification2}, $\tilde h$ actually lies in $L_2(W,A,X) + L_2(Z_{-\myalpha^*},X)$ and satisfies \eqref{outcome.bridge.moment.U2}.
In other words, the required regularity conditions imply the existence of a conventional outcome confounding bridge, which can be decomposed into the sum of two functions, a function of $(W,A,X)$ and a function of $(Z_{-\myalpha^*},X)$, excluding any additive interaction between $Z_{-\myalpha^*}$ and $(W,A)$.

\begin{remark}
Our nonparametric identification results under Assumption~\ref{assump:lower.bound} can be interpreted as identification within a union of viable causal models, each corresponding to a subset $\myalpha^* \subseteq [K]$ satisfying $|\myalpha^*| \geq \gamma$.
In this respect, our approach is conceptually similar to that of \citet{sun2023semiparametric}, who developed a g-estimation framework for identifying the ATE under a semiparametric structural mean model with additive treatment effect \citep{holland1988causal, small2007sensitivity}, in the presence of unmeasured confounding. Their method assumes that at least $\gamma$ of $K$ candidate instruments are valid, without requiring prior knowledge of which instruments are invalid.
Like our approach, their framework enables identification and semiparametric theory without relying on model selection, thereby avoiding assumptions such as the majority rule \citep{han2008detecting, kang2016instrumental} or the plurality rule \citep{guo2018confidence}.
Despite this similarity, our setup presents distinct challenges from \cite{sun2023semiparametric}, in part because proximal inference must integrate both treatment and outcome confounding proxies, and therefore a treatment confounding proxy may be invalid either because it has a direct causal effect on the primary outcome, or because it has a causal effect on an outcome confounding proxy, while a candidate instrument in \cite{sun2023semiparametric} can only be invalid in a single way, due to a direct effect on the primary outcome. 
In addition, \cite{sun2023semiparametric}'s results largely rely on an important assumption that, whether valid or not, the set of candidate instruments must be both mutually independent and jointly independent of $U$. Neither of these conditions is required of valid or invalid treatment confounding proxies in our framework. Finally, they require homogeneity of the causal effect with respect to the hidden confounders, which we also do not require.
 
These fundamental differences call for new technical and methodological developments that distinguish our contributions from existing work on robust instrumental variables and robust proximal methods.


\end{remark}

\subsubsection{Identification via fortified treatment confounding bridge functions}
\label{sec:treatment.bridge}

In this Section, we provide an alternative proximal identification result to that of the previous section using a proximal analog to the inverse probability weighting strategy.
We first introduce a treatment confounding bridge function under Assumption~\ref{assump:lower.bound}.
\begin{assumption}
\label{assump:q}
    There exists a function $q^*$ satisfying $(-1)^{1-A}q^* \in \HH_\gamma$ and
\#
\EE\{q^*(Z,A,X)|W,A,X\} = \frac{1}{\PP(A | W, X)}. \label{treatment.bridge.moment} 
\#
\end{assumption}

We refer to $q^*$ as a \textit{fortified treatment confounding bridge function} to distinguish it from the conventional treatment confounding bridge function in PCI literature \citep{cui2024}.
When $\gamma = K$, then the condition $(-1)^{1-A}q^* \in \HH_K$ holds by definition since any function $q^*$ satisfying \eqref{treatment.bridge.moment} satisfies $\EE\{(-1)^{1-A}q^*(Z,A,X)|X\} = 0$.
Consequently, in this case, the definition recovers that of a conventional treatment confounding bridge function \citep{cui2024}.

We introduce the following regularity condition to identify $\tau^*$ via a fortified treatment confounding bridge function.

\begin{assumption}[Completeness]
\label{assump:treatment.bridge.completeness}
For any $g \in L_2(U, A, Z_{-\myalpha^*}, X)$, 
\$
\EE\{h(W,A,Z_{-\myalpha^*}, X)g(U, A, Z_{-\myalpha^*}, X)\} = 0 ~~~\mbox{ for all }~  h \in L_2(W,A,X) + L_2(Z_{-\myalpha^*}, X)
\$ 
almost surely if and only if $g(U, A, Z_{-\myalpha^*}, X) = 0$ almost surely.
\end{assumption}

Similar to Assumption~\ref{assump:outcome.bridge.completeness}, Assumption~\ref{assump:treatment.bridge.completeness} characterizes the cost of identifying $\tau^*$ without knowing $\myalpha^*$.
To see this, in the oracle setup with $\myalpha^*$ known, Assumption~\ref{assump:valid.completeness}(ii) requires that $L_2(W,A,Z_{-\myalpha^*},X)$ is sufficiently informative about $L_2(U,A,Z_{-\myalpha^*},X)$.
In contrast, Assumption~\ref{assump:treatment.bridge.completeness} requires that $L_2(W,A,X) + L_2(Z_{-\myalpha^*}, X)$, a subspace of $L_2(W,A,Z_{-\myalpha^*},X)$, is sufficiently informative about $L_2(U,A,Z_{-\myalpha^*},X)$.
Thus, Assumption~\ref{assump:treatment.bridge.completeness} is stronger than Assumption~\ref{assump:valid.completeness}(ii), reflecting the potential cost of identifying $\tau^*$ without a priori knowledge of $\myalpha^*$.

We are now ready to establish the identification of $\tau^*$ using fortified treatment confounding bridge functions.

\begin{theorem} \label{thm:treatment.bridge.identification}
    Suppose that Assumptions~\ref{assump:lower.bound}, \ref{assump:q} and \ref{assump:treatment.bridge.completeness} hold.
    Then, any solution $q^*$ to \eqref{treatment.bridge.moment} with $(-1)^{1-A}q^* \in \HH_\gamma$ satisfies
    \#
\EE\{q(Z,A,X)|U, A, Z_{-\myalpha^*}, X\} = \frac{1}{\PP(A|U, Z_{-\myalpha^*}, X)} \label{treatment.bridge.moment.U}
    \#
    almost surely.
    Furthermore, provided that Assumption~\ref{assump:consistency} also holds, the ATE is nonparametrically identified by 
    \#
    \tau^* = \EE\{q^*(Z, 1, X)AY - q^*(Z, 0, X)(1-A)Y\}. \label{treatment.identification}
    \#
\end{theorem}
Theorem~\ref{thm:treatment.bridge.identification} provides an alternative identification result, complementing Theorem~\ref{thm:outcome.bridge.identification}.  
Similarly to Theorem~\ref{thm:outcome.bridge.identification}, Theorem~\ref{thm:treatment.bridge.identification} is therefore a multiply robust causal identification result, as it does not require explicit knowledge of the valid treatment confounding proxies $Z_{\myalpha^*}$.
Additionally, Theorem~\ref{thm:treatment.bridge.identification} shows that any solution $q$ solving \eqref{treatment.bridge.moment} with $(-1)^{1-A}q \in \HH_\gamma$ identifies $\tau^*$ without necessitating the uniqueness of the solution.

\begin{remark} \label{remark:alternative.outcome}
Instead of taking a solution to \eqref{outcome.bridge.moment} as a starting point, we can adopt an alternative identifying condition that posits the existence of a function $h^\prime(w,a,x)$ satisfying \eqref{outcome.bridge.moment.U} almost surely.
Then, one can show that any solution $h^\prime$ also identifies $\tau^*$ via \eqref{outcome.identification} and satisfies \eqref{outcome.bridge.moment}. 
Furthermore, replacing Assumption~\ref{assump:outcome.bridge.completeness} with an appropriate completeness assumption that does not involve unobservable variables and unknown valid treatment confounding proxies ensures the uniqueness of $\tilde h$. 
Consequently, under this setup, the solution of \eqref{outcome.bridge.moment} also identifies $\tau^*$ when substituted into \eqref{outcome.identification}.
An analogous alternative existence and completeness framework applies to the fortified treatment confounding bridge function.
For further details, we refer the reader to Theorems A.1 and A.2 in Appendix A.2 of the supplementary material.
\end{remark}

\section{Semiparametric theory and inference} \label{sec:3}

\subsection{Semiparametric local efficiency bound} \label{sec:semiparametric}

In this Section, we study efficient and robust estimation for $\tau^*$ in the presence of potentially invalid treatment confounding proxies.
For a fixed $\gamma \in [K]$, let $\cM_\gamma$ denote the collection of regular laws of the observed data, which only assumes the existence of a solution $h^*$ to \eqref{outcome.bridge.moment} and is otherwise unrestricted.
For a formal definition of a regular law, we refer readers to Chapter 3 of \cite{bickel1993efficient}. 
Let $O = (Y,A,Z,W,X)$  denote the observed data.

To characterize the semiparametric local efficiency bound for $\tau^*$ under model $\cM_\gamma$, we need the following surjectivity condition.

\begin{assumption}[Surjectivity] \label{assump:surjectivity} 
    For a fixed $\gamma \in [K]$, let $\Pi_\gamma : L_2(W,A,X) \to \HH_\gamma$ denote the operator defined by $\Pi_\gamma(h) = \argmin_{d \in \HH_\gamma}\|h - d\|_2$ for $h \in L_2(W,A,X)$.
    Then, $\Pi_\gamma$ is surjective.
\end{assumption}

We note that $\HH_\gamma \subseteq L_2(O)$ is itself a closed subspace, as it is the intersection of finite closed subspaces.
Consequently, the orthogonal projection operator from $L_2(O)$ to $\HH_\gamma$ exists uniquely (see, e.g., Appendix A.2 in \cite{bickel1993efficient}) and the operator $\Pi_\gamma$ is also well-defined, as it can be obtained by restricting the domain of the orthogonal projection operator to $L_2(W,A,X)$.
With a slight abuse of notation, let $\Pi_\gamma$ also denote the orthogonal projection from $L_2(O)$ to $\HH_\gamma$.
Then, the definition of $h^*$ in \eqref{outcome.bridge.moment} is equivalent to the condition $\Pi_\gamma(h^*) = \Pi_\gamma(Y)$. 
Thus, identifying $h^*$ corresponds to solving an inverse problem involving the linear operator $\Pi_\gamma$.

Recall that $h^*$ solves \eqref{outcome.bridge.moment} if and only if there exists $\ell^* \in \cS_\gamma$ such that $\tilde h = h^* + \ell^*$ satisfies \eqref{outcome.bridge.moment2}.
The following theorem identifies an influence function for $\tau^*$ under model $\cM_\gamma$, and derives a semiparametric local efficiency bound for $\cM_\gamma$ at a specific submodel, $\cM_{{\rm eff}, \gamma}$ defined below.

\begin{theorem} \label{thm:influence function}
Suppose that Assumptions~\ref{assump:consistency} and \ref{assump:lower.bound} -- \ref{assump:treatment.bridge.completeness} hold.
Let $h^* \in L_2(W,A,X)$ and $\ell^* \in \cS_\gamma$ be functions such that $\tilde h = h^* + \ell^*$ satisfies \eqref{outcome.bridge.moment2} and $q^*$ be a function satisfying \eqref{treatment.bridge.moment} with $(-1)^{1-A}q^* \in \HH_\gamma$.
\begin{enumerate}
\item An influence function for $\tau^*$ under model $\cM_\gamma$ is given by
\#
& \mathtt{IF}(O;h^*, q^*, \ell^*,\tau^*) = \{h^*(W,1,X) - h^*(W,0,X)\} \nn \\
& + (-1)^{1-A}q^*(Z, A, X)\{Y - h^*(W,A,X) - \ell^*(Z,X)\} - \tau^*. \label{influence.function}
\#
\item The efficient influence function for $\tau^*$ under the semiparametric model $\cM_{\gamma}$, at the submodel $\cM_{{\rm eff}, \gamma}$,
where Assumption~\ref{assump:surjectivity} holds and $h^* + \ell^*$ and $q^*$ are uniquely defined, is given by $\mathtt{IF}(O;h^*, q^*, \ell^*,\tau^*)$.
In particular, the semiparametric local efficiency bound of $\cM_{\gamma}$ at $\cM_{{\rm eff}, \gamma}$ is given by $\sigma_{{\rm ATE}}^2 = {\rm Var}(\mathtt{IF}(O;h^*, q^*, \ell^*,\tau^*))$.
\end{enumerate}
    
\end{theorem}

The influence function $\mathtt{IF}$ resembles that of $\tau^*$ in PCI under the assumption of known, valid treatment confounding proxies \citep{cui2024}. 
In particular, due to the additive form of $\tilde h = h^* + \ell^*$, the influence function in \eqref{influence.function} can be equivalently expressed as
\$
\tilde h(W,1,Z,X) - \tilde h(W,0,Z,X) + (-1)^{1-A}q^*(Z,A,X)\{Y - \tilde h(W,A,Z,X)\} - \tau^*.
\$
By Proposition~\ref{prop:outcome.bridge.identification2}, $\tilde h$ lies in $L_2(W,A,X) + L_2(Z_{-\myalpha^*},X)$ and satisfies \eqref{outcome.bridge.moment2}, implying $\tilde h$ can be considered as a conventional outcome confounding bridge function \citep{miao2024confounding} when $Z_{\myalpha^*}$ and $Z_{-\myalpha^*}$ are viewed as treatment confounding proxies and baseline covariates, respectively.
This observation further illustrates the similarity between our influence function and that of \cite{cui2024}. 
However, a key distinction is that $\tilde h$ in our framework is constructed without knowledge of $\myalpha^*$, whereas the conventional outcome confounding bridge function assumes $\myalpha^*$ is known.

Theorem~\ref{thm:influence function} shows that the estimator $\tilde \tau$ defined by the estimating equation $0 = \PP_n(\mathtt{IF}(O;h^*, q^*, \ell^*,\tilde \tau))$ attains the semiparametric local efficiency bound of $\cM_{\gamma}$ at $\cM_{{\rm eff}, \gamma}$, where $\PP_n$ is the empirical average operator, i.e. $\PP_n(V) = n^{-1}\sn V_i$ for any random variable $V$.
However, this estimator is infeasible in practice since it depends on the unknown nuisance functions $h^*, q^*$, and $\ell^*$.
Note that the definitions of fortified bridge functions in \eqref{outcome.bridge.moment} and \eqref{treatment.bridge.moment} involve solving Fredholm integral equations of the first kind, which are well-known to often be ill-posed problems.
To illustrate this, for a fixed $\gamma \in [K]$, recall the definition of $\Pi_\gamma$ in Assumption~\ref{assump:surjectivity} and let $\Pi_\gamma^*: \HH_\gamma \to L_2(W,A,X)$ denote the conditional expectation given by $\Pi_\gamma^*(d) = \EE\{d(Z,A,X)|W,A,X\}$ for $d \in \HH_\gamma$.
As previously discussed, the definition of $h^*$ in \eqref{outcome.bridge.moment} is equivalent to the condition $\Pi_\gamma(h^*) = \Pi_\gamma(Y)$, and similarly, the definition of $q^*$ in \eqref{treatment.bridge.moment} corresponds to the condition $\Pi^*_\gamma((-1)^{1-A}q^*) = (-1)^{1-A}/\PP(A|W,X)$.
Thus, identifying these functions reduces to solving these inverse problems.
Even if the operators are injective under suitable regularity conditions ensuring the existence and uniqueness of a solution to these equations, the problems may still be effectively ill-posed.
For instance, if $\Pi_\gamma$ (or $\Pi^*_\gamma$) is a compact operator and zero is a limit point of its eigenvalues, as can sometimes happen in practice, then the inverse operator $\Pi_\gamma^{-1}$ (or $(\Pi^*_\gamma)^{-1}$) may become discontinuous. 
This discontinuity can make the resulting solutions highly sensitive to small perturbations in the observed data due to sampling variability, leading to instability \citep{o1986statistical}.

In the context of conventional PCI, several approaches have been introduced to construct nonparametric estimators of bridge functions (see, e.g. \cite{singh2020kernel, kallus2021causal, mastouri2021proximal, ghassami2022minimax}).
However, as the subspace $\HH_\gamma$ depends on the unknown conditional law $f(Z,A \mid X)$, directly extending these nonparametric approaches to our setting presents significant new challenges.
To address this issue, we adopt a practical semiparametric modeling strategy that effectively regularizes the problem.
Similar to the approach of \cite{cui2024}, we proceed by a priori restricting unknown nuisance functions to follow certain smooth parametric working models, yet still allowing the observed data distribution to otherwise remain unrestricted.
In Section~\ref{sec:estimation.inference}, we introduce a rich class of multiply robust estimators for $\tau^*$ when the parameters of these models are estimated appropriately, without prior knowledge of the valid treatment confounding proxies.

\subsection{Estimation and inference of the average treatment effect} \label{sec:estimation.inference}

We introduce three strategies based on working submodels for the nuisance functions.
We assume that the conditional law $f(Z, A \mid X; \ba)$ is a prespecified parametric model, smooth in $\ba$, and that $h(W, A, X; \bb)$, $\ell(Z, X; \br)$, and $q(Z, A, X; \bt)$ are known parametric functions, each smooth in its respective parameter. 
We impose the constraints $h(0,0,X;\bb) = 0$ for all $\bb$ and $\ell(Z, X; \br) \in \cS_\gamma$ for all $\br$.
The first strategy specifies a semiparametric model denoted $\cM_{1, \gamma}$, which posits  parametric models for $f^*(Z,A|X) = f(Z,A|X;\ba^*)$ and $h(W,A,X) = h(W,A,X;\bb^*)$, where $(\ba^*, \bb^*)$ are unknown finite-dimensional parameters and $\ba$ and $\bb$ are variation-independent.
The second strategy specifies a semiparametric model $\cM_{2, \gamma}$, which posits a parametric model for $h(W,A,X) = h(W,A,X;\bb^*)$ and $\ell^*(Z,X) = \ell(Z,X;\br^*)$, where $(\bb^*, \br^*)$ are unknown finite-dimensional parameters and $\bb$ and $\br$ are variation-independent.
Finally, the third strategy specifies a semiparametric model $\cM_{3, \gamma}$, which posits a parametric model for $f^*(Z,A|X) = f(Z,A|X;\ba^*)$ and $q^*(Z,A,X) = q(Z,A,X;\bt^*)$, where $(\ba^*, \bt^*)$ are unknown finite-dimensional parameters and $\ba$ and $\bt$ are variation-independent.

Let $\hat\ba$ satisfy the estimating equation
\$
0 = \frac{1}{n} \sn \frac{\partial}{\partial \ba} \ln f(Z_i, A_i|X_i;\ba)\bigg|_{\ba = \hat \ba},
\$
so that $\hat\ba$ is the maximum likelihood estimator (MLE) under models $\cM_{1, \gamma}$ and $\cM_{3, \gamma}$.
Define $\HH_\gamma(\hat\ba)$ analogously to $\HH_\gamma$, but with all expectations evaluated under the law $f(Z, A \mid X; \hat\ba)$ in \eqref{def:H.gamma}, and let $d(\cdot;\hat\ba)$ be a mapping from $L_2(Z,A,X)$ to $\HH_\gamma(\hat\ba)$ that depends on $f(Z,A|X;\hat\ba)$. 
We present details on appropriate choices of $d(\cdot;\hat\ba)$ following Theorem~\ref{thm:mr.estimator} below.

Let $\mathbf{c}_0(\hat\ba)$ be a vector-valued function with the same dimension as $\bb$, where each element lies in $\HH_\gamma(\hat\ba)$ via the mapping $d(\cdot; \hat\ba)$.
Define $\hat\bb = \hat\bb(\mathbf{c}_0(\hat\ba))$ as the solution to
\#
0 = \frac{1}{n} \sn \mathbf{c}_0(Z_i, A_i, X_i;\hat\ba) \{Y_i - h(W_i, A_i, X_i;\bb) - \ell(Z_i, X_i;\hat\br(\bb))\}, \label{def:h.estimates}
\#
where $\hat\br(\bb)$ satisfies
\$
0 = \frac{1}{n}\sn \mathbf{c}_1(Z_i, A_i, X_i)\{Y_i - h(W_i, A_i, X_i;\bb) - \ell(Z_i, X_i;\br)\},
\$
where $\mathbf{c}_1(Z_i, A_i, X_i)$ is a vector function of the same dimension as $\br$.
Finally, to construct an estimator for $\bt^*$, we note that, because $q^*$ satisfies  \eqref{treatment.bridge.moment}, 
\$
\EE\{(-1)^{1-A}q^*(Z,A,X)c(W,A,X)\} & = \EE\bigg\{\frac{(-1)^{1-A}}{\PP(A|W,X)}c(W,A,X)\bigg\} \\
& = c(W,1,X) - c(W,0,X)
\$
for any function $c$, with $(-1)^{1-A}q^* \in \HH_\gamma$.
Motivated by this, for any vector-valued function $\mathbf{c}_2$ of the same dimension as $\bt$, define $\hat\bt$ as the solution to
\$
0 = \frac{1}{n} \sn \{\mathbf{c}_2(W_i, A_i, X_i) d((-1)^{1 - A}q(\bt);\hat\ba))(Z_i, A_i, X_i) - \mathbf{c}_2(W_i, 1, X_i) + \mathbf{c}_2(W_i, 0, X_i)\},
\$
where $q(\bt) = q(Z,A,X;\bt)$.

Based on the identification results established in Sections~\ref{sec:outcome.bridge} and~\ref{sec:treatment.bridge}, we define the \textit{fortified proximal outcome regression} (fPOR) estimator $\hat\tau_{\rm fPOR}$ and the \textit{fortified proximal inverse probability weighted} (fPIPW) estimator $\hat\tau_{\rm fPIPW}$ as follows:
\#
& \hat \tau_{{\rm fPOR}} := \PP_n\{h(W,1,X;\hat \bb) - h(W,0,X;\hat\bb)\},  \label{def:rpor} \\
& \hat \tau_{{\rm fPIPW}} := \PP_n\{d((-1)^{1-A}q(\hat\bt);\hat\ba)(Z,A,X)Y\}. \label{def:rpipw}
\#

The estimators $\hat\bb$ (under $\cM_{1, \gamma}$ or $\cM_{2, \gamma}$) and $\hat\bt$ (under $\cM_{3, \gamma}$) are consistent and asymptotically normal (CAN).
Consequently, $\hat \tau_{{\rm fPOR}}$ provides valid inference for $\tau^*$ if either $\cM_{1, \gamma}$ or $\cM_{2, \gamma}$ holds, and $\hat\tau_{{\rm fPIPW}}$ provides valid inference if $\cM_{3, \gamma}$ holds.

However, $\hat\tau_{{\rm fPOR}}$ and $\hat\tau_{{\rm fPIPW}}$ may be severely biased and yield invalid inference if their corresponding models, $\cM_{1, \gamma}$ and $\cM_{2, \gamma}$ or $\cM_{3, \gamma}$, respectively, are misspecified.
To mitigate this, we propose the following \textit{fortified proximal multiply robust} (fPMR) estimator:
\#
& \hat \tau_{{\rm fPMR}} = \PP_n[h(W,1,X;\hat\bb) - h(W,0,X;\hat\bb) \nn \\
& ~~~~~~~~~~~~~~~+ d((-1)^{1-A}q(\hat\bt);\hat\ba)(Z,A,X)\{Y - h(W,A,X;\hat \bb) - \ell(Z_i, X_i;\hat\br)\}]. \label{def:mr.estimator}
\#
The next theorem establishes that, under the union model $\cM_{1, \gamma} \cup \cM_{2, \gamma} \cup \cM_{3, \gamma}$, $\hat \tau_{{\rm fPMR}}$ is multiply robust and locally semiparametric efficient.

\begin{theorem}\label{thm:mr.estimator}
Suppose that the regularity conditions given in the supplementary material hold and that $\ba, \bb, \br$ and $\bt$ are variation-independent.
Then, the estimator $\hat \tau_{{\rm fPMR}}$ defined in \eqref{def:mr.estimator} is a consistent and asymptotically normal estimator of $\tau^*$ under the semiparametric union model $\cM_{1, \gamma} \cup \cM_{2, \gamma} \cup \cM_{3, \gamma}$. Furthermore, $\hat \tau_{{\rm fPMR}}$ is semiparametrically locally efficient under the semiparametric union model, at the intersection submodel $\cM_{{\rm eff}, \gamma} \cap \cM_{1, \gamma} \cap \cM_{2,\gamma} \cap \cM_{3, \gamma}$.
\end{theorem}
Theorem~\ref{thm:mr.estimator} proves that our proposed estimator is multiply robust without requiring correct specification of a set of valid treatment confounding proxies. 
It delivers valid inference provided that at least one of the following holds: the conditional distribution of $(Z,A)$ given $X$, and the fortified outcome confounding bridge function are correctly specified; both the fortified outcome confounding bridge function and the conditional expectation of a corresponding residual are correctly specified; or the conditional distribution of $(Z,A)$ given $X$ and the fortified treatment confounding bridge function are correctly specified.

We now describe two separate approaches for constructing functions in $\HH_\gamma(\hat\ba)$, each based on a specific mapping $d(\cdot; \hat\ba)$ from $L_2(Z, A, X)$ to $\HH_\gamma(\hat\ba)$.
The following lemma generalizes Theorem 1 of \citet{tchetgen2010doubly} and provides an explicit representation of $\HH_\gamma$.

\begin{lemma} \label{lem:change.of.measure}
Let $f^*(Z, A| X)$ be any fixed user-specified conditional law of $(A, Z)$ given $X$ such that $(A,Z)$ are conditionally independent given $X$ under the law $f^*$. Moreover, suppose that $f^*$ is absolutely continuous with respect to the true law $f(Z, A|X)$ and $f/f^* \in L_2(Z, A, X)$.
Then, the set of functions $\HH_\gamma$ is $\{d^\dagger(g)(Z,A,X) : g \in L_2(Z,A,X)\}$, where
     \$
     d^\dagger(g)(Z,A,X) = \frac{f^*}{f}\bigg\{g - \sum_{i = \gamma}^K \sum_{\myalpha \in \cP_{i}([K])} \alpha_i\EE^*(g|Z_{-\myalpha}, X)\bigg\}.
     \$
     Here, the coefficients $\{\alpha_i\}_{i = \gamma}^K$ are explicitly specified in the proof, and the conditional expectations $\EE^*(\cdot|\cdot)$ are taken with respect to the probability law $f^*$.
\end{lemma}

Alternatively, since $\HH_\gamma$ is the intersection of the closed subspaces, any function $d^{(0)}(Z,A,X) \in L_2(Z,A,X)$ can be mapped to an element in $\HH_\gamma$ using the alternating conditional expectations (ACE) algorithm \citep{breiman1985estimating, bickel1993efficient}. 
Let $\myalpha_1, \ldots, \myalpha_M$ enumerate all subsets of $[K]$ of size $\gamma$, where $M = \binom{K}{\gamma}$. 
Starting from $d^{(0)}$, the ACE algorithm iteratively computes the sequence
\#
d^{(M\cdot m + j)}(Z,A,X) = d^{(M \cdot m + j - 1)}(Z,A,X) - \EE\{d^{(M \cdot m + j - 1)}(Z,A,X)|Z_{-\myalpha_j}, X\} \label{ace.algorithm}
\#
for $1 \leq j \leq M$ and $m \geq 1$.
This iterative procedure converges to a function $d^\ddag(d^{(0)})(Z,A,X)$ in $\HH_\gamma$ as $m \to \infty$.

Given a function $d^{(0)} \in L_2(Z,A,X)$, let $d^\dagger(d^{(0)}; \hat\ba)(Z, A, X)$ and $d^\ddag(d^{(0)}; \hat\ba)(Z, A, X)$ denote the functions $d^\dagger(d^{(0)})(Z,A,X)$ and $d^\ddag(d^{(0)})(Z,A,X)$, respectively, with expectations taken under the law $f(Z, A \mid X; \hat\ba)$. 
Then, both $d^\dagger$ and $d^\ddag$ define valid mappings from $L_2(Z, A, X)$ to $\HH_\gamma(\hat\ba)$.

\begin{remark}
When constructing a function in $\HH_\gamma(\hat\ba)$, the mapping $d^\dagger$ involves weighting by the inverse of the joint conditional law $f(Z,A \mid X;\hat\ba)$, 
which can lead to instability if the estimated law is either inaccurate or reaches values near zero. 
To mitigate this, we use the $d^\ddag$ mapping implemented via the ACE algorithm in both simulations and data analysis.
Following the approach of \citet{breiman1985estimating} and \citet{vansteelandt2008multiply}, we adopt a practical strategy that specifies separate working models for the conditional expectations in equation~\eqref{ace.algorithm}, fitting each iteratively using standard regression methods.
A limitation of this approach is that the functions obtained via the ACE algorithm lack closed-form expressions, preventing direct computation of the asymptotic variance through derivatives of the estimating equations. 
Consequently, following \citet{vansteelandt2008multiply}, we employ the nonparametric bootstrap to conduct inference under this practical implementation of the ACE algorithm.
\end{remark}

\section{Simulation} \label{sec:simulation}

In this section, we perform Monte Carlo experiments to assess the finite-sample performance of the proposed estimators.
We first describe the data-generating process (DGP).
We consider the candidate treatment confounding proxies $Z = (Z_1, Z_2)$ with $K = 2$.
We generated $X, U, Z_2$ by the multivariate normal distribution,
\$
\left(\begin{array}{c}
     X  \\
     Z_2  \\
     U
\end{array} \right)  \sim 
{\rm MVN} \left( \left( \begin{array}{c}
     0  \\
     0  \\
     0
\end{array} \right), \left( \begin{array}{cccc}
     1 & 1/2 & 1/2  \\
     1/2 & 1 & 1/2  \\
     1/2 & 1/2 & 1 
\end{array}\right) \right).
\$
We then generate $A$ and $Z_1$ as follows.
\$
& \PP(A = 1 | X, Z_2, U) = \{1 + \exp(-2(U - 0.5 X))\}^{-1}, \\
& Z_1 \sim \mathcal{N}(2(-1)^{1-A}(U - 0.5 X), 0.5^2).
\$
Finally, $W$ and $Y$ are generated as follows:
\$
W \sim \mathcal{N}(-1 - X - 0.5 U + 0.25 Z_2, 0.5^2), \\
Y \sim \mathcal{N}(1 + X + 0.25 U + \tau^*A - 0.25 Z_2, 0.25^2).
\$

From the definition of the DGP, $Z_1$ is the valid treatment confounding proxy, while $Z_2$ is invalid.
Moreover, in Appendix A.4 of the supplementary material, we show that the above DGP is compatible with the following parametric models of $h,\ell$ and $q$:
\#
h(W,A,X;\bb) &= b_a A + b_w W, \label{h.para}\\
\ell(Z,X;\br) &= r_0 + r_z^\T Z + r_x X, \label{ell.para} \\
q(Z,A,X;\bt) &=  1 + \exp(t_0 + t_z^\T Z + t_a A + t_x X), \nn
\#
where $b_a = \tau^*, b_w = -1/2$, $r_0 = 1/2, r_z = (0, -1/8)^\T, r_x = 1/2$ and $t_0 = -1/8, t_z = (-1, 0)^\T, t_a = 0, t_x = 0$.
We choose $\tau^* = 2$.

We implemented our proposed estimators: fortified proximal analogs of outcome regression estimator $\hat \tau_{{\rm fPOR}}$ defined in  \eqref{def:rpor}, fortified proximal inverse probability weighted estimator $\hat \tau_{{\rm fPIPW}}$ defined in \eqref{def:rpipw}, and fortified proximal multiply robust estimator $\hat \tau_{{\rm fPMR}}$.
We applied the ACE algorithm under the practical strategy described in Section~\ref{sec:estimation.inference} using linear regression models for the conditional expectations in \eqref{ace.algorithm}.
Standard errors were obtained via the ordinary nonparametric bootstrap with 500 bootstrap samples, excluding the most extreme $1\%$ of estimates to improve stability.

\begin{table}
  \centering
  \caption{Simulation results for various estimators: absolute value of the empirical bias (Bias) ($\times 10^{-2}$) and standard deviation (SD) ($\times 10^{-2}$) of the estimators, average of the estimated standard errors (SE) ($\times 10^{-2}$), and coverage of $95\%$ confidence intervals (Cov).}
  \begin{tabular}{lrrrr}
    \midrule
    {Method} & {Bias} & {SD} & {SE} & {Cov} \\
    \midrule
    $\hat \tau_{{\rm fPOR}, 1}$ & 0.7 & 5.8 & 8.6 & 97.2 \\
    $\hat \tau_{{\rm fPIPW}}$ & 1.0 & 4.9 & 3.5 & 92.3 \\
    $\hat \tau_{{\rm fPMR}}$ & 0.0 & 3.3 & 3.0 & 96.2 \\
    $\hat \tau_{{\rm PDR}, 1}$ & 0.1 & 2.7 & 2.5 & 96.3 \\
    $\hat \tau_{{\rm PDR}, 2}$ & 26.0 & 22.1 & 16.1 & 45.2 \\
    $\hat \tau_{{\rm DR}}$ & 15.6 & 1.1 & 1.2 & 0\\
    \bottomrule
  \end{tabular}
  \label{table1}
\end{table}

For comparative purposes, we also implemented several benchmark estimators, including two conventional proximal estimators that assume knowledge of valid treatment confounding proxies \citep{cui2024}, as well as a standard doubly robust estimator that ignores unmeasured confounding and treats all of $Z$, $W$, and $X$ as baseline covariates.
The conventional proximal estimators, denoted by $\hat{\tau}_{{\rm PDR}, j}$ for $j = 1, 2$, are proximal analogs of doubly robust estimators that use $Z_j$ as the treatment confounding proxy while incorporating the remaining proxy along with $X$ as baseline covariates. 
Among these, $\hat{\tau}_{{\rm PDR}, 1}$ is regarded as the oracle estimator since it relies on a correctly specified valid proxy. 
In contrast, $\hat{\tau}_{{\rm PDR}, 2}$ incorrectly treats the invalid proxy as if it were valid.
The standard doubly robust estimator is given by
\$
\hat \tau_{{\rm DR}} = \PP_n\bigg[ \hat \EE(Y|L, A = 1) - \hat \EE(Y|L, A = 0) + \frac{(-1)^{1-A}}{\hat f(A|L)}\{Y - \hat E(Y|L, A)\} \bigg],
\$
where $\hat{f}(A \mid L)$ denotes an estimator of the conditional law of $A$ given $L$, and $\hat{\EE}(Y \mid L, A)$ is an estimator of the conditional expectation of $Y$ given $L$ and $A$. 
These quantities were estimated via standard logistic and linear regression, respectively.

We considered a sample size of $n = 3000$ over $1000$ Monte Carlo samples. 
Table~\ref{table1} summarizes the simulation results. 
As expected, the oracle estimator $\hat{\tau}_{{\rm PDR}, 1}$ exhibited negligible bias and the lowest standard deviation among the six estimators evaluated. 
Since all parametric models were correctly specified, our proposed estimators achieved small biases and coverage probabilities close to the nominal level, even in the absence of knowledge about the valid proxy, which is consistent with our theoretical guarantees. 
In contrast, $\hat{\tau}_{{\rm PDR}, 2}$ suffered from substantial biases and under-coverage, attributable to the use of invalid treatment confounding proxies. 
Similarly, the standard doubly robust estimator $\hat \tau_{{\rm DR}}$ is severely biased due to unmeasured confounding.

Motivated by \cite{kang2007demystifying}, we evaluate the robustness of our proposed estimators under model misspecification by considering scenarios in which one or both confounding bridge functions are incorrectly specified. 
Specifically, we consider the following three scenarios:
\begin{enumerate}
    \item[I.] We used $W^* = W\cdot |X|^{1/2}$ instead of $W$ for parametric modeling of $h$, that is, $h$ is misspecified ($\cM_{1, \gamma}$ and $\cM_{2, \gamma}$ are misspecified).
    \item[II.] We used $Z^* = (Z_1^2, Z_2^2)^\T$ instead of $Z$ for parametric modeling of $\ell$, and $Z^{**} = (Z_1\cdot X, Z_2 \cdot X)^\T$ for parametric modeling of $q$, that is, $\ell$ and $q$ are misspecified ($\cM_{2, \gamma}$ and $\cM_{3, \gamma}$ are misspecified).
    \item[III.] We used $Z^{**}$ for parametric modeling of $q$, and chose $d(d^{(0)};\hat\ba)(Z,A,X) = d^{(0)}$, that is, $q$ is misspecified and the ACE algorithm converges to incorrectly specified functions ($\cM_{1, \gamma}$ and $\cM_{3, \gamma}$ are misspecified).
\end{enumerate}

\begin{table}[t]
  \centering
  \caption{Simulation results under various model misspecification scenarios (Scen.): absolute value of the empirical bias (Bias) ($\times 10^{-2}$) and standard deviation (SD) ($\times 10^{-2}$) of the estimators, average of the estimated standard errors (SE) ($\times 10^{-2}$), and coverage of $95\%$ confidence intervals (Cov).}
  \begin{tabular}{clrrrr}
    \midrule
    {Scen.} & {Method} & {Bias} & {SD} & {SE} & {Cov} \\
    \midrule
    \multirow{3}{*}{I}
    & $\hat \tau_{{\rm fPOR}}$ & 22.9 & 76.4 & 36.6 & 84.2\\
    & $\hat \tau_{{\rm fPIPW}}$ & 1.0 & 4.9 & 3.5 & 92.3\\
    & $\hat \tau_{{\rm fPMR}}$ & 0.8 & 21.8 & 7.9 & 94.9 \\
    \midrule
    \multirow{3}{*}{II}
    & $\hat \tau_{{\rm fPOR}}$ & 0.8 & 5.9 & 9.0 & 97.3\\
    & $\hat \tau_{{\rm fPIPW}}$ & 14.4 & 1.9 & 2.0 & 0.2\\
    & $\hat \tau_{{\rm fPMR}}$ & 0.8 & 5.9 & 9.0 & 97.4 \\
    \midrule
    \multirow{3}{*}{III}
    & $\hat \tau_{{\rm fPOR}}$ & 0.7 & 5.8 & 8.6 & 97.2\\
    & $\hat \tau_{{\rm fPIPW}}$ & 16.4 & 8.2 & 8.4 & 48.7\\
    & $\hat \tau_{{\rm fPMR}}$ & 0.7 & 5.8 & 8.6 & 97.5 \\
    \bottomrule
  \end{tabular}
  \label{table2}
\end{table}

Table~\ref{table2} presents simulation results for the three misspecification scenarios under a sample size of $n = 3000$, based on $1000$ repetitions.
As expected, the multiply robust estimator $\hat \tau_{{\rm fPMR}}$ exhibits small bias and achieves coverage probabilities close to the nominal level in Scenarios I -- III, that is, when at least one of the parametric models is correctly specified.
This confirms its multiply robust property.
In contrast, the other two estimators exhibit substantial bias when the corresponding models $\cM_{1, \gamma}, \cM_{2, \gamma}$ and $\cM_{3, \gamma}$ are misspecified.
These findings provide empirical support for our theoretical results in finite samples and underscore the practical benefits of the proposed multiply robust estimator.

\section{Data Analysis}
\label{sec:data.analysis}

Right heart catheterization (RHC) is a diagnostic procedure used to assess cardiac function in critically ill patients. 
We apply our proposed method to data from the Study to Understand Prognoses and Preferences for Outcomes and Risks of Treatments (SUPPORT) to evaluate the causal effect of RHC in the intensive care unit (ICU)  \citep{connors1996effectiveness}.
The SUPPORT dataset includes 5,735 hospitalized patients, of whom 2,184 received RHC within the first 24 hours of study entry ($A_i = 1$), and 3,551 did not ($A_i = 0$). 
The outcome of interest is the number of days from admission to either death or censoring at 30 days. 
A total of 71 baseline demographic, clinical, and laboratory variables were collected; see Table 1 of \cite{hirano2001estimation} for details.

Ten physiological variables, obtained from blood tests conducted during the first 24 hours in the ICU, can be regarded as potential proxies for unmeasured confounding.
Following \cite{tchetgen2024anint} and \cite{cui2024}, we designate $\texttt{ph1}$ and $\texttt{hema1}$ as outcome confounding proxies ($W = (\texttt{ph1}, \texttt{hema1})$). 
The other eight physiological variables are treated as candidate treatment confounding proxies, $Z = $(\texttt{sod1}, \texttt{pot1}, \texttt{crea1}, \texttt{bili1}, \texttt{alb1}, \texttt{pafi1}, \texttt{paco21}, \texttt{wblc1}) with $K = 8$, while the remaining 61 variables are included as baseline covariates $X$.
We specified the fortified outcome confounding bridge function and the conditional expectation of its residual as in equations~\eqref{h.para} and \eqref{ell.para}, respectively, incorporating additional interactions among up to $K - \gamma$ components of $Z$ for $\gamma \in \{2,4,6,8\}$. The model for the fortified treatment confounding bridge function was specified as
\$
q(Z,A,X;\bt) &=  1 + \exp\left((-1)^{1-A}(t_0 + t_z^\T Z + t_a A + t_x X)\right).
\$

For comparison, we also consider conventional proximal outcome confounding bridge function estimator ($\hat \tau_{{\rm POR}}$), proximal treatment confounding bridge function estimator ($\hat \tau_{{\rm PIPW}}$), proximal doubly robust estimator ($\hat\tau_{{\rm PDR}}$) with allocating \texttt{pafi1} and \texttt{paco21} as valid treatment confounding proxies \citep{cui2024}, and the standard doubly robust estimator ($\hat \tau_{{\rm DR}}$).

\begin{table}[t]
  \centering
  \caption{Point estimates and 95\% confidence intervals (in parentheses) of the causal effects of RHC on survival time up to 30 days.
  The top panel reports our proposed estimators indexed by $\gamma$ with $K = 8$, while the bottom panel shows the doubly robust estimator and conventional proximal estimators in \cite{cui2024}.}
  \begin{tabular}{cccc}
    \midrule
    $\gamma$ & $\hat \tau_{{\rm fPOR}}$ & $\hat \tau_{{\rm fPIPW}}$ & $\hat \tau_{{\rm fPMR}}$  \\
    \midrule
    2 & -0.75 (-10.23, 8.74) & -1.27 (-2.18, -0.35)   & -1.25 (-2.17, -0.34)  \\
    4 & -2.20 (-19.93, 12.49) & -1.28 (-2.14, -0.42)  & -1.28 (-2.14, -0.42) \\
    6 & -2.25 (-18.98, 14.48) & -1.42 (-2.31, -0.52)  & -1.41 (-2.30, -0.51) \\
    8 & -1.60 (-4.37, 1.17) & -1.20 (-1.86, -0.53)  & -1.20 (-4.43, 2.04)\\
    \midrule
    \midrule
    $\hat \tau_{{\rm POR}}$ & $\hat \tau_{{\rm PIPW}}$ & $\hat \tau_{{\rm PDR}}$ & $\hat \tau_{{\rm DR}}$ \\
    \midrule
    -1.80 (-2.65, -0.94) & -1.72 (-2.30, -1.14) & -1.66 (-2.50, -0.83) & -1.17 (-1.79, -0.55) \\
    \bottomrule
  \end{tabular}
  \label{table3}
\end{table}

Table~\ref{table3} presents the data analysis results. 
Notably, the estimated treatment effects from most of our proposed methods are larger in magnitude than those from the standard doubly robust estimator. 
This pattern generally confirms the findings in \cite{tchetgen2024anint} and \cite{cui2024}, suggesting that RHC may have a more detrimental effect on 30-day survival among critically ill patients admitted into an ICU than previously indicated by using the standard doubly robust estimator.
At the same time, our estimates, $\hat\tau_{{\rm fPIPW}}$ and $\hat\tau_{{\rm fPMR}}$ are more conservative in magnitude than those obtained from the conventional proximal estimators. This reflects the potential robustness benefits of our identifying assumptions: unlike the conventional methods, which assume that both \texttt{pafi1} and \texttt{paco21}, are valid treatment confounding proxies, our approach requires only that at least $\gamma$ out of the 8 candidate proxies are valid without needing to specify which ones are valid.
Finally, we observe that $\hat{\tau}_{{\rm fPOR}}$ is more variable than the other estimators, possibly due to increased variability associated with estimating the required function $\mathbf{c}_0$ in \eqref{def:h.estimates}.

\section{Discussion}

In this paper, we developed a general semiparametric framework for causal inference in the presence of unmeasured confounding by leveraging possibly invalid treatment confounding proxies through the introduction of new classes of confounding bridge functions, which generalize the standard confounding bridge functions of \citep{miao2024confounding} and \cite{cui2024}. We proposed a new class of multiply robust estimators for the ATE, which are locally efficient. The utility of our approach was demonstrated through simulation studies and a real data application.

There are several promising directions for future research. While our framework accommodates invalid treatment confounding proxies, it still relies critically on the validity of the outcome confounding proxies, in the sense that $W$ and $A$ must be conditionally independent given $U$, $X$, and $Z_{-\myalpha^*}$. A natural extension is to develop identification strategies for the ATE when some of the candidate outcome confounding proxies are invalid, or when both treatment and outcome confounding proxies may be partially invalid; we plan to further explore these strategies in future work.

Another important direction is to explore the use of modern machine learning methods to estimate the fortified confounding bridge functions, as alternatives to positing parametric working models. In the context of conventional PCI, adversarial inference frameworks incorporating reproducing kernel Hilbert spaces (RKHSs) or neural networks have shown promise \citep{singh2020kernel, kallus2021causal, mastouri2021proximal, ghassami2022minimax}. However, extending these approaches to our setting is nontrivial, as the structure of $\HH_\gamma$ is not readily approximated by conventional function classes such as RKHSs or neural networks.

\appendix

\section{Details of the main paper}

\subsection{Assumption~\ref{assump:h} when \texorpdfstring{$\gamma = K$}{gamma = K}}
\label{sec:gamma.K}

Let $\gamma = K$ and suppose that Assumption~\ref{assump:h} holds, that is, there exists $h \in L_2(W,A,X)$ satisfying
\$
\EE[d(Z,A,X)\{Y - h(W,A,X)\}] = 0 ~~~\forall d \in \HH_K.
\$
Recall that when $\gamma = K$,
\$
\HH_K = \{d(Z,A,X) \in L_2(Z,A,X) : \EE\{d(Z,A,X)|X\} = 0\}.
\$
Denote $\ell(X) = \EE\{Y - h(W,A,X)|X\}$. 
Since we have
\$
\EE\{d(Z,A,X)\ell(X)\} = \EE[\EE\{d(Z,A,X)|X\}\ell(X)] = 0 ~~\forall d \in \HH_K,
\$
it follows that
\$
\EE[d(Z,A,X)\{Y - h(W,A,X) - \ell(X)\}] = 0 ~~~\forall d \in \HH_K.
\$
Choose $d(Z,A,X) = \EE\{Y - h(W,A,X)|Z,A,X\} - \ell(X)$.
By definition of $\ell$, we have $d \in \HH_K$, implying
\$
0 & = \EE[d(Z,A,X)\{Y - h(W,A,X) - \ell(X)\}] \\
& = \EE[\EE\{Y - h(W,A,X) - \ell(X)|Z,A,X\}\{Y - h(W,A,X) - \ell(X)\}] \\
& = \EE[\EE\{Y - h(W,A,X) - \ell(X)|Z,A,X\}^2].
\$
Therefore, $\EE\{Y - h(W,A,X) - \ell(X)|Z,A,X\} = 0$ almost surely, and the conventional outcome confounding bridge function exists with $\tilde h(W,A,X) = h(W,A,X) + \ell(X)$.

Conversely, if there exists a conventional outcome confounding bridge function $\tilde h$ satisfying $\EE\{Y - \tilde h(W,A,X)|Z,A,X\} = 0$, then it is straightforward that $\tilde h$ also satisfies \eqref{outcome.bridge.moment} as $\HH_K \subseteq L_2(Z,A,X)$.

\subsection{Alternative identification results} \label{sec:alternative.identification}

\begin{theorem} \label{thm:alternative.outcome.identification}
Suppose that Assumptions~\ref{assump:consistency} and \ref{assump:lower.bound} hold and that there exists a function $h^\prime(w,a,x)$ satisfying 
\$
\EE\{Y - h^\prime(W,A,X)|U,A, Z_{-\myalpha^*}, X\} = \EE\{Y - h^\prime(W,A,X)|Z_{-\myalpha^*}, X\},
\$
almost surely.
Then, we have
\#
\tau^* = \EE\{h^\prime(W,1,X) - h^\prime(W,0,X)\},
\# \label{outcome.identification.appendix}
and $h^\prime$ also solves \eqref{outcome.bridge.moment}.
Furthermore, suppose that for any square-integrable function $g \in L_2(W,A,X)$, $\EE\{d(Z,A,X)g(W,A,X)\} = 0$ for all $d \in \HH_\gamma$ if and only if $g = \EE(g|X)$ almost surely.
Then, $\tau^*$ is identified by plugging any solution to \eqref{outcome.bridge.moment} into \eqref{outcome.identification.appendix}.
\end{theorem}

\begin{theorem} \label{thm:alternative.treatment.identification}
    Suppose that Assumptions~\ref{assump:consistency} and \ref{assump:lower.bound} hold, and that there exists a function $q^\prime(z,a,x)$ satisfying $(-1)^{1-A}q^\prime \in \HH_\gamma$ and 
    \$
\EE\{q^\prime(Z,A,X)|U,A,Z_{-\myalpha^*},X\} = \frac{1}{\PP(A|U,Z_{-\myalpha^*},X)}
    \$
    almost surely.
    Then, we have
    \#
\tau^* = \EE\{q^\prime(Z,1,X)AY - q^\prime(Z,0,X)(1-A)Y\}, \label{treatment.identification.appendix}
    \#
    and $q^\prime$ also solves \eqref{treatment.bridge.moment}.
    Furthermore, suppose that for any square-integrable function $g \in \HH_\gamma$, $\EE\{h(W,A,X)g(Z,A,X)\} = 0$ for all $h \in L_2(W,A,X)$ if and only if $g = 0$ almost surely. Then, any function $q$ satisfying \eqref{treatment.bridge.moment} and $(-1)^{1-A}q \in \HH_\gamma$ is unique. Consequently, $\tau^*$ is identified by substituting this solution into \eqref{treatment.identification.appendix}.
\end{theorem}

\begin{proof}[Proof of Theorem~\ref{thm:alternative.outcome.identification}]

In the proof of Theorem~\ref{thm:outcome.bridge.identification}, it is proved that such $h^\prime$ satisfies \eqref{outcome.identification.appendix}, so it suffices to prove that $h^\prime$ also solves \eqref{outcome.bridge.moment}.
To show this, denote $\EE\{Y - h^\prime(W,A,X)|Z_{-\myalpha^*},X\} = \ell^\prime(Z_{-\myalpha^*},X)$ and note that $Y \perp Z_{\myalpha^*}|U,A,Z_{-\myalpha^*},X$ and $W \perp (Z_{\myalpha^*},A)|U,Z_{-\myalpha^*},X$.
Then, by the Law of Iterative Expectation, we have
\$
& \EE\{Y - h^\prime(W,A,X) - \ell^\prime(Z_{-\myalpha^*},X)|Z,A,X\} \\
& = \EE[\EE\{Y - h^\prime(W,A,X) - \ell^\prime(Z_{-\myalpha^*},X)|U,Z,A,X\}|Z,A,X\} \\
& = \EE[\EE\{Y - h^\prime(W,A,X) - \ell^\prime(Z_{-\myalpha^*},X)|U,Z_{-\myalpha^*},A,X\}|Z,A,X\} \\
& = 0,
\$
where the last equation follows from the definition of $h^\prime$.
Thus, we have
\$
\EE\{Y - h^\prime(W,A,X)|Z,A,X\} = \ell^\prime(Z_{-\myalpha^*},X).
\$
For any $d \in \HH_\gamma$, note that
\$
\EE[d(Z,A,X)\{Y - h^\prime(W,A,X)\}] & = \EE\{d(Z,A,X)\ell^\prime(Z_{-\myalpha^*},X)\} \\
& = \EE[\EE\{d(Z,A,X)|Z_{-\myalpha^*},X\}\ell^\prime(Z_{-\myalpha^*},X)] = 0.
\$
Thus, $h^\prime$ solves the class of moment equations \eqref{outcome.bridge.moment}.

Now, let $h$ be any solution to \eqref{outcome.bridge.moment}.
Since $h$ and $h^\prime$ are solutions to \eqref{outcome.bridge.moment}, we have
\$
\EE[d(Z,A,X)\{h(W,A,X) - h^\prime(W,A,X)\}] = 0 ~~\forall d \in \HH_\gamma.
\$
Then, by the completeness assumption, it follows that
\$
h^\prime(W,A,X) - h(W,A,X) = \EE\{h^\prime(W,A,X) - h(W,A,X)|X\} =: \ell^\prime(X).
\$
Therefore,
\$
\tau^* = \EE\{h^\prime(W,1,X) - h^\prime(W,0,X)\} & = \EE[h(W,1,X) + \ell^\prime(X) - \{h(W,0,X) + \ell^\prime(X)\}] \\
& = \EE\{h(W,1,X) - h(W,0,X)\},
\$
establishing the claim.
\end{proof}

\begin{proof}[Proof of Theorem~\ref{thm:alternative.treatment.identification}]
In the proof of Theorem~\ref{thm:treatment.bridge.identification}, we have shown that such $q^\prime$ satisfies \eqref{treatment.identification.appendix}.
Thus, we only need to show that $q^\prime$ also solves \eqref{treatment.bridge.moment}.
Recall that $W \perp (Z_{\myalpha^*},A)|U,Z_{-\myalpha^*},X$.
Then, by the Law of Iterative Expectation, we have
\$
& \EE\{q^\prime(Z,A,X)|W,A,X\} \\
& = \EE[\EE\{q^\prime(Z,A,X)|U,W,A,Z_{-\myalpha^*}, X\}|W,A,X] \\
& = \EE[\EE\{q^\prime(Z,A,X)|U,A,Z_{-\myalpha^*}, X\}|W,A,X] \\
& = \EE\bigg\{ \frac{1}{\PP(A|U,Z_{-\myalpha^*},X)} \bigg| W, A, X \bigg\} \\
& = \EE\bigg\{ \frac{1}{\PP(A|U,W, Z_{-\myalpha^*},X)} \bigg| W, A, X \bigg\} = \frac{1}{\PP(A|W, X)}.
\$
Thus, $q^\prime$ solves \eqref{treatment.bridge.moment}.

Now, to show the uniqueness, suppose that both $q(Z,A,X)$ and $q^\prime(Z,A,X)$ satisfy \eqref{treatment.bridge.moment} and $(-1)^{1-A}q(Z,A,X), (-1)^{1-A}q^\prime(Z,A,X) \in \HH_\gamma$.
Then, we must have
\$
\EE[(-1)^{1-A}\{q(Z,A,X) - q^\prime(Z,A,X)\}|W,A,X] = 0.
\$
Then, by the completeness assumption, $(-1)^{1-A}q = (-1)^{1-A}q^\prime$, that is, $q = q^\prime$ almost surely.
Thus, the solution $q^\prime$ to \eqref{treatment.bridge.moment} with $(-1)^{1-A}q^\prime \in \HH_\gamma$ is unique and the treatment confounding bridge function is identified from \eqref{treatment.bridge.moment}.
    
\end{proof}

\subsection{Sufficient conditions for the existence of the bridge functions} \label{appendix:sufficient.condition}

By following a similar approach as in \cite{miao2018identifying}, we provide sufficient conditions that ensure the existence of a solution to \eqref{outcome.bridge.moment}.
Sufficient conditions that guarantee the existence of a fortified treatment confounding bridge function can be established in a similar fashion and are omitted.
We first state Picard's theorem, which gives a necessary and sufficient condition for the existence of the solution to the first-kind equations.
\begin{theorem} [Theorem 15.18 of \cite{kress1989linear}]
    Let $A: X \to Y$ be a compact linear operator with singular system $(\mu_n, \varphi_n, g_n)_{n = 1}^\infty$.
    The equation of the first kind 
    \$
A\varphi = f
    \$
    is solvable if and only if $f \in N(A^*)^\perp = \{f^\prime : A^*(f^\prime) = 0\}^\perp $ and $\sum_{n = 1}^\infty \nu_n^{-2}|\langle f, g_n \rangle|^2 < \infty$.
\end{theorem}

We first show that the existence of a solution to \eqref{outcome.bridge.moment} is equivalent to the existence of solutions to some non-homogeneous equations.
For each $x$, let $L_2(W,A|X = x)$ and $L_2(Z,A|X = x)$ be the spaces of square-integrable functions of $(W,A)$ and $(Z,A)$ given $X = x$, respectively, which are equipped with the $L_2$ inner product. 
Similarly, let $L_2(Y,W,Z,A|X = x)$ be the space of square-integrable functions of $(Y,W,Z,A)$ given $X = x$.
Define
\$
\HH_{\gamma, x} := \{d \in L_2(Z,A|X = x) : \EE(d|Z_{-\myalpha}, X = x) = 0 ~\forall \myalpha \in \cP_{\gamma}([K])\},
\$
and let $\Pi_{\gamma, x} : L_2(Y,W,Z,A|X = x) \to \HH_{\gamma, x}$ be the orthogonal projection onto $\HH_{\gamma, x}$.
\begin{lemma} \label{lem:sol.equivalence}
    There exists a solution to \eqref{outcome.bridge.moment} if and only if for each $X = x$, there exists a solution to the non-homogeneous equation $\Pi_{\gamma, x}(h) = \Pi_{\gamma, x}(Y)$ for $h \in L_2(W,A|X = x)$.
\end{lemma}

\begin{proof} 
    First, we assume that there exists a solution to \eqref{outcome.bridge.moment}. 
    This means that $\Pi_{\gamma}(Y - h^*) = 0$, that is, $\Pi_{\gamma}(h^*) = \Pi_{\gamma}(Y)$.
    Choose any function $d \in \HH_{\gamma, x}$.
    Then, it can be shown that $\mathbbm{1}(X = x)d(Z,A) \in \HH_{\gamma}$ so that for any $d \in \HH_{\gamma, x}$, we have
    \$
\EE[\mathbbm{1}(X = x)d(Z,A)\{Y - h^*(W,A,X)\}] = 0,
    \$
    which implies
    \$
\EE[d(Z,A)\{Y - h^*(W,A,x)\}|X = x] = 0.
    \$
    Since this equation holds for all $d \in \HH_{\gamma, x}$, $h^*(W,A,x)$ solves the non-homogeneous equation $\Pi_{\gamma, x}(h) = \Pi_{\gamma, x}(Y)$.

    Conversely, assume that for each $x$, there exists $h_x \in L_2(W,A|X = x)$ satisfying $\Pi_{\gamma, x}(h) = \Pi_{\gamma, x}(Y)$.
    Let $h^* \in L_2(Z,A,X)$ be given by $h^*(W,A,x) = h_x(W,A)$ for each $x$.
    Now, for any $d \in \HH_\gamma, \myalpha \in \cP_{\gamma}([K])$ and $x$, $\EE\{d(Z,A,X)|Z_{-\myalpha}, X = x\} = 0$,
    implying that $d(Z,A,x) \in \HH_{\gamma, x}$.
    Thus, 
    \$
\EE[d(Z,A,X)\{h^*(W, A, X) - Y\}|X = x] & = \EE[d(Z,A,x)\{h_x(W,A) - Y\}|X = x] \\
& = 0.
    \$
    This equality holds for any $d \in \HH_\gamma$ and $x$, implying that $h^*$ satisfies \eqref{outcome.bridge.moment}, completing the proof.
\end{proof}

Abusing some notations, denote $\Pi_{\gamma, x}$ to be the operator by restriction the domain of $\Pi_{\gamma, x}$ to $L_2(W,A|X = x)$ for each $x$.
Then, by Lemma~\ref{lem:sol.equivalence}, we need to show that $\Pi_{\gamma, x}(h) = \Pi_{\gamma, x}(Y)$ is solvable. 
To show this, we assume the following conditions for each $x$:

\begin{enumerate}
    \item[(i)] $\sum_{a \in \{0,1\}}\int \int f_{W|Z, A, X}(w|Z, a, x)f_{Z|W,A,X}(Z|w,a,x){\rm d}w {\rm d}Z < \infty$;
    \item[(ii)] For $g \in \HH_{\gamma, x}$, $\EE\{g(Z, A)|W, A, X = x\} = 0$ implies $g = 0$ almost surely;
    \item[(iii)] $\sum_{i = 1}^\infty \mu_{i, x}^{-2}|\langle g_{Y, x}(Z, A), g_{i, x} \rangle_{L_2(Z, A|X = x}|^2 < \infty$, where $\{\mu_{i, x}, \phi_{i, x}, g_{i, x}\}_{i = 1}^\infty$ is the singular system of $\Pi_{\gamma, x}$ and $g_{Y, x} = \Pi_{\gamma, x}(Y)$.
\end{enumerate}

We first show that under the above assumptions, the operator $\Pi_{\gamma, x}: L_2(W,A|X = x) \to \HH_{\gamma, x}$ is compact.
To see this, let $\Pi_{\gamma, x}^* : \HH_{\gamma, x} \to L_2(W,A|X = x)$ be the conditional expectation operator defined as follows:
\$
\Pi_{\gamma, x}^*(g) = \EE\{g(Z,A)|W, A, X = x\}.
\$
Then, $\Pi_{\gamma, x}^*$ is the adjoint operator of $\Pi_{\gamma, x}$. 
To verify this, for any $h \in L_2(W,A|X = x)$ and $g \in \HH_{\gamma, x}$,
\$
\langle \Pi_{\gamma, x}(h), g \rangle_{\HH_{\gamma, x}} & = \EE\{ \Pi_{\gamma, x}(h) g(Z,A)|X = x\} = \EE\{ h(W,A) g(Z,A)|X = x\} \\
& = \EE[h(W,A)\EE\{g(Z,A)|W,A,X = x\}|X = x] \\
& = \langle h, \Pi_{\gamma, x}^*(g) \rangle_{L_2(W,A|X = x)},
\$
which proves the claim.
Moreover, we remark that $\Pi_{\gamma, x}^*$ is the restriction of the conditional expectation operator $\Pi_x^*: L_2(Z, A|X = x) \to L_2(W, A|X = x)$ to $\HH_{\gamma, x}$.
Therefore, $\Pi_{\gamma, x}$ is compact if and only if $\Pi_{\gamma, x}^*$ is compact, and $\Pi_{\gamma, x}^*$ is compact given that $\Pi_{x}^*$ is compact (see, e.g., Theorem 4.4 in Chapter II of \cite{conway1990acourse}).
To show that $\Pi_x^*$ is compact, note that for any $(w,a)$,
\$
\Pi_x^*(g)(w,a) & = \EE\{g(Z,A)|W = w, A = a, X = x\} \\
& = \sum_{a^\prime \in \{0,1\}} \int g(Z, a^\prime)f(Z, a^\prime|x)\frac{\mathbbm{I}(a = a^\prime)f(w, Z, a|x)}{f(Z,a^\prime|x)f(w,a|x)} {\rm d}Z \\
& = \langle k(\cdot, (w, a)), g \rangle_{L_2(Z,A|X = x)},
\$
where the kernel $k$ is defined as follows:
\$
k((Z,a^\prime), (w,a)) = \frac{\mathbbm{I}(a = a^\prime)f(w,Z,a|x)}{f(w,a^\prime|x)f(Z,a|x)}.
\$
Therefore, as shown in page 5659 of \cite{carrasco2007linear}, $\Pi_x^*$ is compact given that
\$
& \sum_{a \in \{0,1\}}\sum_{a^\prime \in \{0,1\}} \int \int \frac{\mathbbm{I}(a = a^\prime)f^2(w,Z,a|x)}{f^2(w,a^\prime|x)f^2(Z,a|x)} f(w,a^\prime|x)f(Z,a|x){\rm d}w {\rm d}Z \\
& = \sum_{a \in \{0,1\}}\int \int f_{W|Z, A, X}(w|Z, a, x)f_{Z|W,A,X}(Z|w,a,x){\rm d}w {\rm d}Z < \infty,
\$
which is assumed in the above.

Next, we show that $N(\Pi_{\gamma, x}^*)^\perp = L_2(Z,A|X = x)$, which ensures that $\Pi_{\gamma, x}(Y)$ belongs to $ N(\Pi_{\gamma, x}^*)^\perp$ trivially. 
To see this, remark that $g \in N(\Pi_{\gamma, x}^*)$ if and only if $\EE\{g(Z,A)|W,A,X = x\} = 0$.
But by the above assumptions, this implies $g = 0$, so that $N(\Pi_{\gamma, x}^*)$ contains only the zero function, leading to $N(\Pi_{\gamma, x}^*)^\perp = L_2(Z,A|X = x)$.
Putting the pieces together with the above assumption on the singular system of $\Pi_{\gamma, x}$, the conditions of Theorem 15.18 of \cite{kress1989linear} are satisfied.
Consequently, for each $x$, there exists a solution $h_x \in L_2(W,A|X =x)$ satisfying $\Pi_{\gamma, x}(h_x) = \Pi_{\gamma, x}(Y)$, which, establishes the existence of a solution to \eqref{outcome.bridge.moment} by Lemma~\ref{lem:sol.equivalence}.

\subsection{Choices of the parameters for data generating process}
\label{appendix:choice.of.parameters}

Let $\myalpha^* = \{1\}$.
By the construction of $W$ and $Y$, it is straightforward that $\myalpha^*$ is the set of indices of valid treatment confounding proxies.

We next present the choices of the parameters of the parametric modeling of $h,\ell$, and $q$ for the data-generating process in Section~\ref{sec:simulation}.

\subsubsection{Parameters for the fortified treatment confounding bridge function}

Recall that conditional on $Z_2, X, U$, $A$ is generated from the following Bernoulli distribution,
\$
\PP(A = 1 | Z_2, X, U) = \{1 + \exp(- 2(U - 0.5 X))\}^{-1}.
\$
Then, $Z_1$ is generated from the univariate normal distribution $Z_1 \sim \mathcal{N}(2(-1)^{1-A}(U - 0.5X), 0.5^2)$.
Then,
\$
\EE\{q(Z,A,X;\bt^*)|U,Z_2,A,X\} & = 1 + e^{-1/8}\EE\{e^{-Z_1}|U,Z_2,A,X\} \\
& = 1 + e^{-1/8}e^{-2(-1)^{1-A}(U- 0.5 X)}\times e^{1/8} = \frac{1}{\PP(A|Z_2, X, U)},
\$
satisfying \eqref{treatment.bridge.moment}.
To see that $(-1)^{1-A}q(\bt^*) \in \HH_\gamma$, we first note that
\$
\EE\{(-1)^{1-A}q(Z,A,X;\bt^*)|Z_2, X\} & = \EE[\EE\{(-1)^{1-A}q(Z,A,X;\bt^*)|Z_2, U, X\}|Z_2, X] \\
& = \EE\bigg[\EE\bigg\{\frac{(-1)^{1-A}}{\PP(A|Z_2, X, U)}\bigg\}\bigg|Z_2, X\bigg] = 0.
\$
Thus, it suffices to show that
\$
\EE\{(-1)^{1-A}q(Z,A,X;\bt^*)|Z_1, X\} = 0.
\$
Remark that $U|X = x \sim \mathcal{N}(0.5x, 3/4)$. 
Now, the joint distribution of $(Z_1, U, X, A)$ is decomposed as
\$
f(Z_1 = z_1, U = u, X = x, A = 1) = f(z_1|u, 1, x)f(1|u,x)f(u,x).
\$
Let $\phi(\cdot)$ be the probability density function of the standard normal distribution.
Then, 
\$
& f(z_1|u,1,x) = 2\phi(2(z_1 - 2(u - 0.5x))), \\
& f(1|u,x) = (1 + e^{-2(u - 0.5x)})^{-1}, \\
& f(u,x) = \phi((u - 0.5x)/\sigma)/\sigma \times \phi(x),
\$
where $\sigma^2 = 3/4$.
Combining these terms and integrating the $f(z_1, u, x, 1)$ over $U$, we have
\$
& f(Z_1 = z_1, X = x, A = 1) \\
& = \frac{2\phi(x)}{\sigma}\int_{-\infty}^\infty \phi(2(z_1 - 2(u - 0.5x)))(1 + e^{-2(u-0.5x)})^{-1}\phi((u - 0.5x)/\sigma){\rm d}u.
\$
Similarly, we have
\$
& f(Z_1 = z_1, X = x, A = 0) \\
& = \frac{2\phi(x)}{\sigma}\int_{-\infty}^\infty \phi(2(z_1 + 2(u - 0.5x)))(1 + e^{2(u-0.5x)})^{-1}\phi((u - 0.5x)/\sigma){\rm d}u.
\$
Applying the change of variable and the fact that $\phi(\cdot)$ is symmetric, we have
\$
f(Z_1 = z_1, X = x, A = 1) = f(Z_1 = z_1, X = x, A = 0),
\$
which implies that
\$
\PP(A = 1|Z_1, X) = \PP(A = 0|Z_1, X) = 1/2
\$
almost surely.
Combining this with our choice of $q(\bt^*)$, it follows that
\$
\EE\{(-1)^{1-A}q(Z,A,X;\bt^*)|Z_1, X\} & = \EE\{q(Z,1,X;\bt^*)|Z_1, A = 1, X\}\PP(A = 1|Z_1, X) \\
& - \EE\{q(Z,0,X;\bt^*)|Z_1, A = 0, X\}\PP(A = 0|Z_1, X) \\
& = (1 + e^{-1/8 - Z_1})\{\PP(A = 1|Z_1, X) - \PP(A = 0|Z_1, X)\} \\
& = 0,
\$ 
which establishes the claim.

\subsubsection{Parameters for the fortified outcome confounding bridge function}

Recall that $W$ and $Y$ are generated as follows:
\$
W \sim \mathcal{N}(-1 - X - 0.5U + 0.25Z_2, 0.5^2), \\
Y \sim \mathcal{N}(1 + X + 0.25 U + \tau^*A - 0.25Z_2, 0.25^2).
\$

Therefore, we have
\$
& \EE(W|A, X, U, Z) = -1 - X - 0.5U + 0.25Z_2, \\
& \EE(Y|A, X, U, Z) = 1 + X + 0.25 U + \tau^*A - 0.25Z_2,
\$
implying
\$
\EE(Y + W/2 - \tau^* A|Z, A, X) = 1/2 + X/2 - Z_2/8.
\$
Denoting $h(W,A,X) = -W/2 + \tau^* A$, this implies that
\$
\EE[d(Z,A,X)\{Y - h(W,A,X)\}] = 0 ~~~\forall d \in \HH_1,
\$
and $h$ is the fortified outcome confounding bridge function, and $\ell(Z,X) = 1/2 + X/2 - Z_2/8$.

\section{Proofs of results from the main text}

\subsection{Proof of Theorem~\ref{thm:outcome.bridge.identification}}

To begin with, we show that $h^*$ satisfying \eqref{outcome.bridge.moment} also satisfies
\#
\EE\{Y - h^*(W,A,X)|U,A,Z_{-\myalpha^*},X\} = \EE\{Y - h^*(W,A,X)|Z_{-\myalpha^*},X\}. \label{outcome.identification.proof1}
\#
To show this, by definition, we have
\$
\EE[d(Z,A,X)\{Y - h^*(W,A,X)\}] = 0
\$
for any $d \in \HH_\gamma$.
Then, by the Law of Iterative Expectation, it follows that
\#
\EE[d(Z,A,X)\EE\{Y - h^*(W,A,X)|U,Z,A,X\}] = 0. \label{outcome.identification.proof2}
\#
Now, since $Y \perp Z_{\myalpha^*}|U,A, Z_{-\myalpha^*},X$ and $W \perp (Z_{\myalpha^*}, A)|U,Z_{-\myalpha^*}, X$, it follows that
\$
\EE\{Y|U, Z, A, X\} = \EE\{Y|U, A, Z_{-\myalpha^*}, X\},
\$
and
\$
\EE\{h^*(W,A,X)|U,Z,A,X\} = \EE\{h^*(W,A,X)|U,Z_{-\myalpha^*}, A, X\}.
\$
Therefore, \eqref{outcome.identification.proof2} implies that for all $d \in \HH_\gamma$,
\$
\EE[d(Z,A,X)\EE\{Y - h^*(W,A,X)|U,Z_{-\myalpha^*},A,X\}] = 0.
\$
Now, let 
\$
\ell^*(Z_{-\myalpha^*},X) = \EE\{Y - h^*(W,A,X)|Z_{-\myalpha^*}, X\}.
\$
For any $d \in \HH_\gamma$,
\$
\EE\{d(Z,A,X)\ell^*(Z_{-\myalpha^*},X)\} & = \EE[\EE\{d(Z,A,X)|Z_{-\myalpha^*},X\}\ell^*(Z_{-\myalpha^*},X)] = 0,
\$
implying that for any $d \in \HH_\gamma$,
\$
\EE[d(Z,A,X)\EE\{Y - h^*(W,A,X) - \ell^*(Z_{-\myalpha*},X) |U,Z_{-\myalpha^*},A,X\}] = 0.
\$
Moreover, for any $d \in L_2(Z_{-\myalpha^*},X)$, we have by definition of $\ell^*$ that
\$
\EE[d(Z_{-\myalpha^*},X)\EE\{Y - h^*(W,A,X) - \ell^*(Z_{-\myalpha*},X) |U,Z_{-\myalpha^*},A,X\}] = 0.
\$
Then, by Assumption~\ref{assump:outcome.bridge.completeness}, we have 
\$
\EE\{Y - h^*(W,A,X) - \ell^*(Z_{-\myalpha^*},X)|U,Z_{-\myalpha^*},A,X\} = 0,
\$
which is equivalent to \eqref{outcome.identification.proof1}.

Now, by Definition~\ref{def:valid.proxy}, it follows for any $a \in \{0,1\}$ that
\$
\EE\{Y(a)|U, Z_{-\myalpha^*}, X\} & = \EE\{Y(a)|U, A = a, Z_{-\myalpha^*}, X\} \\
& = \EE\{Y|U, A = a, Z_{-\myalpha^*}, X\} \\
& = \EE\{h^*(W,A,X) + \ell^*(Z_{-\myalpha^*},X)|U, A = a, Z_{-\myalpha^*}, X\} \\
& = \EE\{h^*(W,a,X) + \ell^*(Z_{-\myalpha^*},X)|U, A = a, Z_{-\myalpha^*}, X\} \\
& = \EE\{h^*(W,a,X) + \ell^*(Z_{-\myalpha^*},X)|U, Z_{-\myalpha^*}, X\},
\$
where the second equality follows by Assumption~\ref{assump:consistency}.
Therefore, we have
\$
\tau^* = \EE\{Y(1)\} - \EE\{Y(0)\} & = \EE[\EE\{Y(1) - Y(0)|U, Z_{-\myalpha^*}, X\}] \\
& = \EE[\EE\{h^*(W,1,X) - h^*(W,0,X)|U, Z_{-\myalpha^*}, X\}] \\
& = \EE\{h^*(W,1,X) - h^*(W,0,X)\},
\$
which establishes the claim.

\qed

\subsection{Proof of Proposition~\ref{prop:outcome.bridge.identification2}}

Remark that
\$
\HH_\gamma = \bigcap_{{\myalpha \in \cP_\gamma([K])}} \{d \in L_2(Z,A,X) : \EE(d|Z_{-\myalpha}, X) = 0\}.
\$
By definition of $h^*$, we have 
\$
\EE[d(Z,A,X)\{Y - h^*(W,A,X)\}] = 0 ~~ \forall d \in \HH_\gamma,
\$
which is equivalent to 
\$
\EE\{Y - h^*(W,A,X) | Z, A, X\} & \in \bigg( \bigcap_{{\myalpha \in \cP_\gamma([K])}} \{d \in L_2(Z,A,X) : \EE(d|Z_{-\myalpha}, X) = 0\} \bigg)^{\perp} \\ 
& = \overline{\bigg(\sum_{\myalpha \in \cP_{\gamma}([K])} \{d \in L_2(Z,A,X): \EE(d|Z_{-\myalpha}, X) = 0 \}^\perp \bigg)},
\$
where $A^\perp$ is the orthogonal complement of a subset $A$. 
Thus, it suffices to identify the orthogonal complement of $\{d \in L_2(Z,A,X) : \EE(d|Z_{-\myalpha}, X) = 0\}$.
To derive this for a fixed $\myalpha$, define the operator $T_\myalpha : L_2(Z,A,X) \to L_2(Z,A,X)$ as $T_\myalpha(d) = \EE(d|Z_{-\myalpha}, X)$.
Since $T_\myalpha$ is an orthogonal projection, $T_\myalpha$ is self-adjoint and $N(T_\myalpha)^\perp = R(T_\myalpha)$, where $N(T_\myalpha)$ and $R(T_\myalpha)$ are the null space and range space of $T_\myalpha$, respectively (see, e.g. Appendix A.2 in \cite{bickel1993efficient}).
It is straightforward to see that $N(T_\myalpha) = \{d \in L_2(Z,A,X) : \EE(d|Z_{-\myalpha}, X) = 0\}$, and $R(T_\myalpha) = L_2(Z_{-\myalpha}, X)$ for each $\myalpha$, implying
\$
\EE\{Y - h^*(W,A,X) | Z, A, X\} \in \overline{\sum_{\myalpha \in \cP_\gamma([K])} L_2(Z_{-\myalpha}, X)} = \cS_\gamma. 
\$
Therefore, if $h^*$ satisfies \eqref{outcome.bridge.moment},
\$
\ell^*(Z,X) := \EE\{Y - h^*(W,A,X)|Z,A,X\} \in \cS_\gamma,
\$
and by its definition, we have 
\$
\EE\{Y - h^*(W,A,X) - \ell^*(Z,X)|Z,A,X\} = 0.
\$
Conversely, assume that there exist $h^* \in L_2(W,A,X)$ and $\ell^* \in \cS_\gamma$ satisfying
\$
\EE\{Y - h^*(W,A,X) - \ell^*(Z,X)|Z,A,X\} = 0.
\$
Since we have $\HH_\gamma^\perp = \cS_\gamma$, it follows that for any $d \in \HH_\gamma$,
\$
0 = \EE[d(Z,A,X)\{Y - h^*(W,A,X) - \ell^*(Z,X)\}] = \EE[d(Z,A,X)\{Y - h^*(W,A,X)\}],
\$
so that $h^*$ satisfies \eqref{outcome.bridge.moment}.
The remaining claims follow directly from Theorem~\ref{thm:outcome.bridge.identification}, completing the proof.

\qed

\subsection{Proof of Theorem~\ref{thm:treatment.bridge.identification}}

To begin with, we establish that 
\$
\EE\{q^*(Z,A,X)|U, A, Z_{-\myalpha^*}, X\} = \frac{1}{\PP(A|U, Z_{-\myalpha^*}, X)}.
\$
To show this, by definition, we have
\$
\EE\{(-1)^{1-A}q^*(Z,A,X)|W, A, X\} = \frac{(-1)^{1-A}}{\PP(A|W,X)}.
\$
By applying the Law of Iterative Expectation, it follows that
\$
\EE[\EE\{(-1)^{1-A}q^*(Z,A,X)|U, W, A, Z_{-\myalpha^*}, X\}|W, A, X] =  \frac{(-1)^{1-A}}{\PP(A|W,X)}.
\$
By Definition~\ref{def:valid.proxy}, $(A, Z_{\myalpha^*}) \perp W | U, Z_{-\myalpha^*}, X$, we have
\$
\EE\{(-1)^{1-A}q^*(Z,A,X)|U,W,A,Z_{-\myalpha^*},X\} = \EE\{(-1)^{1-A}q^*(Z,A,X)|U,A,Z_{-\myalpha^*},X\},
\$
and
\$
\EE\bigg\{\frac{(-1)^{1-A}}{\PP(A|U, Z_{-\myalpha^*}, X)}\bigg|W, A, X\bigg\} = \EE\bigg\{\frac{(-1)^{1-A}}{\PP(A|U, W, Z_{-\myalpha^*}, X)}\bigg|W, A, X\bigg\} = \frac{(-1)^{1-A}}{\PP(A|W,X)}.
\$
Therefore,
\#
\EE\bigg[\EE\{(-1)^{1-A}q^*(Z,A,X)|U, A, Z_{-\myalpha^*}, X\} - \frac{(-1)^{1-A}}{\PP(A|U, Z_{-\myalpha^*}, X)}\bigg|W, A, X\bigg] = 0. \label{treatment.identification.proof1}
\#
Now, remark that $(-1)^{1-A}q^* \in \HH_\gamma$ and $|\myalpha^*| \geq \gamma$, implying that 
\$
\EE\{(-1)^{1-A}q^*(Z,A,X)|Z_{-\myalpha^*},X\} = 0.
\$
Since $\EE\{(-1)^{1-A}/\PP(A|U,Z_{-\myalpha^*},X)|Z_{-\myalpha^*}, X\} = 0$, we have
\$
\EE\bigg[\EE\{(-1)^{1-A}q^*(Z,A,X)|U, A, Z_{-\myalpha^*}, X\} - \frac{(-1)^{1-A}}{\PP(A|U, Z_{-\myalpha^*}, X)}\bigg| Z_{-\myalpha^*}, X\bigg] = 0.
\$
Combining this with \eqref{treatment.identification.proof1} and Assumption~\ref{assump:treatment.bridge.completeness}, we have
\$
\EE\{(-1)^{1-A}q^*(Z,A,X)|U, A, Z_{-\myalpha^*}, X\} = \frac{(-1)^{1-A}}{\PP(A|U, Z_{-\myalpha^*}, X)},
\$
implying
\$
\EE\{q^*(Z,A,X)|U, A, Z_{-\myalpha^*}, X\} = \frac{1}{\PP(A|U, Z_{-\myalpha^*}, X)}.
\$

Now, remark that
\$
& \EE\{Y(-1)^{1-A} q^*(Z, A, X)|U, Z_{-\myalpha^*}, X\} \\
& = \EE\{Yq^*(Z,1,X)|U, A = 1, Z_{-\myalpha^*}, X\}\PP(A = 1|U, Z_{-\myalpha^*}, X) \\
& ~~~~~~~~~~~~~~~~~~~~- \EE\{Yq^*(Z,0,X)|U, A = 0, Z_{-\myalpha^*}, X\}\PP(A = 0|U, Z_{-\myalpha^*}, X) \\
& = \EE\{Y|U, A = 1, Z_{-\myalpha^*}, X\}\EE\{q^*(Z,1,X)|U, A = 1, Z_{-\myalpha^*}, X\}\PP(A = 1|U, Z_{-\myalpha^*}, X) \\
& - \EE\{Y|U, A = 0, Z_{-\myalpha^*}, X\}\EE\{q^*(Z,0,X)|U, A = 0, Z_{-\myalpha^*}, X\}\PP(A = 0|U, Z_{-\myalpha^*}, X) \\
& = \EE\{Y|U,A = 1,Z_{-\myalpha^*},X\} - \EE\{Y|U,A = 0,Z_{-\myalpha^*},X\} \\
& = \EE\{Y(1)|U,A = 1,Z_{-\myalpha^*},X\} - \EE\{Y(0)|U,A = 0,Z_{-\myalpha^*},X\} \\
& = \EE\{Y(1)|U,Z_{-\myalpha^*},X\} - \EE\{Y(0)|U,Z_{-\myalpha^*},X\}.
\$
Thus, 
\$
\tau^* = \EE\{Y(1) - Y(0)\} = \EE\{Y (-1)^{1-A}q^*(Z, A, X)\},
\$
which completes the proof.

\qed

\subsection{Proof of Theorem~\ref{thm:influence function}}
\subsubsection{Proof of Theorem~\ref{thm:influence function}(i)}
Let $f_t(y,w,z, a, x)$ be a regular parametric submodel with $f_0$ the true density function with respect to an appropriate dominating measure $\mu$. 
The associated score function and the expectation are denoted as $S(y,w,z, a, x) = \partial \log f_t(y,w,z, a, x)/\partial t |_{t = 0}$ and $\EE_t$, respectively.
Moreover, we can similarly denote corresponding ATE, nuisance functions, and score functions for any component of the submodel.
In order to find the efficient influence function of $\tau^*$, we need to find a random variable $G$ with mean $0$ and 
$$
\frac{\partial \tau_t}{\partial t}\bigg|_{t = 0} = \EE\{G \cdot S(Y,W,Z,A,X)\},
$$
where $\tau_t|_{t = 0} = \tau^*$ and $\tau_t$ is the ATE under $f_t$.

To this end, we first recall that $h_t(W, A, X)$ satisfies $h_t(W,A,X)|_{t = 0} = h^*(W,A,X)$ and
\begin{align}
\EE_t[d_t(Z,A,X)\{Y - h_t(W,A,X)\}] = 0, \label{submodel.h}
\end{align}
where $d_t(Z, A, X)$ is any function in $L_2(Z,A,X)$ satisfying
\#
\EE_t\{d_t(Z, A, X)|Z_{-\myalpha}, X\} = 0 \label{submodel.h.gamma}
\#
for any $\myalpha \in \cP_\gamma([K])$.
By taking partial derivatives in \eqref{submodel.h}, we have
$$
\int \frac{\partial [d_t(Z,a,x)\{y - h_t(w,a,x)\}f_t(y,w,Z,a,x)]}{\partial t}\bigg|_{t = 0}\mu({\rm d}y, {\rm d}w, {\rm d}z, {\rm d}a, {\rm d}x) = 0.
$$
Thus, for any $d_t \in L_2(Z, A, X)$ satisfying \eqref{submodel.h.gamma} with $d_0 = d \in \HH_\gamma$, 
\begin{align}
& \EE\bigg\{d(Z,A,X)\epsilon S(Y,W,Z,A,X) + \frac{\partial d_t}{\partial t}(Z,A,X)\bigg|_{t = 0} \cdot \epsilon\bigg\} \nn \\
& ~~~~~~~~~~~~~~~~~~~~~~~~~~~~~~~~~~~~~~~~= \EE\bigg\{d(Z,A,X)\frac{\partial h_t}{\partial t}(W,A,X)\bigg|_{t = 0}\bigg\}, \label{h.derivative}
\end{align}
where $\epsilon = Y - h^*(W,A,X)$.

To compute $\EE(\partial d_t/\partial t|_{t = 0}\cdot\epsilon)$, in the following, we find a function $\tilde d(Z, A, X)$ such that $\tilde d + \partial d_t/\partial t|_{t = 0} \in \HH_\gamma$.
Then, noting $\epsilon = Y - h^*(W,A,X)$ which satisfies \eqref{submodel.h}, it follows that 
\begin{align}
\EE\bigg\{\frac{\partial d_t}{\partial t}(Z, A, X)\bigg|_{t = 0}\cdot\epsilon\bigg\} = - \EE\{\tilde d(Z, A, X)\epsilon\}. \label{d.derivative}
\end{align}
Recall that $d_t$ satisfies \eqref{submodel.h.gamma}.
Thus, by taking the partial derivative at $t = 0$, we have that for any $\myalpha \in \cP_{\gamma}([K])$,
\#
\EE\bigg\{ d(Z, A, X)S(Z_\myalpha, A|Z_{-\myalpha}, X) + \frac{\partial d_t}{\partial t}(Z, A, X)\bigg| Z_{-\myalpha}, X\bigg\} = 0. \label{d.derivative2}
\#
Now, by the definition, we have 
\$
S(Z_\myalpha, A| Z_{-\myalpha}, X) = S(Z, A, X) - \EE\{S(Z, A, X)|Z_{-\myalpha}, X\}.
\$
Thus, for any $\myalpha \in \cP_{\gamma}([K])$,
\$
& \EE\{d(Z, A, X)S(Z_\myalpha, A|Z_{-\myalpha}, X)|Z_{-\myalpha}, X\} \\
& = \EE\{d(Z, A, X)[S(Z, A, X) - \EE\{S(Z, A, X)|Z_{-\myalpha}, X\}]|Z_{-\myalpha}, X\} \\
& = \EE\{d(Z, A, X)S(Z, A, X)|Z_{-\myalpha}, X\} - \EE\{S(Z, A, X)|Z_{-\myalpha}, X\}\EE\{d(Z, A, X)|Z_{-\myalpha}, X\} \\
& = \EE\{[d(Z, A, X) - \EE\{d(Z, A, X)|Z_{-\myalpha}, X\}]S(Z, A, X)|Z_{-\myalpha}, X\} \\
& = \EE\{d(Z, A, X)S(Z, A, X)|Z_{-\myalpha}, X\},
\$
where the last equality follows by the fact that $d \in \HH_\gamma$.
Combining this with \eqref{d.derivative2} yields
\$
& \EE\bigg\{ d(Z, A, X) S(Z, A, X) + \frac{\partial d_t}{\partial t}(Z, A, X)\bigg| Z_{-\myalpha}, X\bigg\} = 0.
\$
Thus, $\tilde d + \partial d_t/\partial t|_{t = 0} \in \HH_\gamma$ where $\tilde d = d\cdot S(Z, A, X)$ and \eqref{d.derivative} implies
\$
\EE\bigg\{\frac{\partial d_t}{\partial t}(Z, A, X)\bigg|_{t = 0}\cdot\epsilon\bigg\} = - \EE\{d(Z, A, X) S(Z, A, X)\epsilon\}.
\$
Together, this and \eqref{h.derivative} give that for all $d \in \HH_\gamma$,
\begin{align}
& \EE[\{d(Z, A, X)S(Y,W,Z,A,X) - d(Z, A, X) S(Z, A, X)\}\epsilon]  \nn \\
& ~~~~~~~~~~~~~~~~~~~~~~~~~~~~~~~~= \EE\bigg\{d(Z,A,X)\frac{\partial h_t}{\partial t}(W,A,X)\bigg|_{t = 0}\bigg\}.\label{h.derivative2}
\end{align}

Now, under the given regularity assumptions, it follows that $\tau_t = \EE_t\{h_t(W,1,X) - h_t(W,0,X)\}$ by Theorem~\ref{thm:outcome.bridge.identification}.
By taking the partial derivative, we have
\#
\frac{\partial \tau_t}{\partial t}\bigg|_{t = 0} & = \EE[\{h^*(W,1,X) - h^*(W,0,X)\}S(Y, W, Z, A, X)]\nn \\
&  ~~~~~~~~+ \EE\bigg\{\frac{\partial h_t(W,1,X)}{\partial t}\bigg|_{t = 0} - \frac{\partial h_t(W,0,X)}{\partial t}\bigg|_{t = 0}\bigg\}.  \label{tau.derivative}
\#
Remark that
\$
\EE\bigg\{\frac{\partial h_t(W,1,X)}{\partial t}\bigg|_{t = 0} - \frac{\partial h_t(W,0,X)}{\partial t}\bigg|_{t = 0}\bigg\} = \EE\bigg\{\frac{(-1)^{1-A}}{\PP(A|W, X)} \frac{\partial h_t(W,A,X)}{\partial t}\bigg|_{t = 0} \bigg\}.
\$
Thus, for any function $q^*(Z, A, X)$ satisfying 
\#
\EE\{q^*(Z, A, X)|W,A,X\} = \frac{1}{\PP(A|W,X)}, \label{q.dagger}
\#
and $(-1)^{1-A}q^* \in \HH_\gamma$, \eqref{tau.derivative} can be written as
\$
& \frac{\partial \tau_t}{\partial t}\bigg|_{t = 0} \\
& = \EE\bigg[\{h^*(W,1,X) - h^*(W,0,X)\}S(Y,W,Z,A,X) + (-1)^{1-A}q^*(Z,A,X)\frac{\partial h_t(W,A,X)}{\partial t}\bigg].
\$
Now, as $(-1)^{1-A}q^* \in \HH_\gamma$, \eqref{h.derivative2} implies
\$
& \EE\bigg\{(-1)^{1-A}q^*(Z,A,X)\frac{\partial h_t(W,A,X)}{\partial t}\bigg\} \\
& = \EE[\{(-1)^{1-A}q^*(Z, A, X)S(Y,W,Z,A,X) - q^*(Z, A, X) S(Z, A, X)\}\epsilon].
\$
By applying the Law of Iterative Expectation,
\$
& \EE\{(-1)^{1-A}q^*(Z,A,X)S(Z,A,X)\epsilon\} \\
& =  \EE\{(-1)^{1-A}q^*(Z,A,X)S(Z,A,X)\EE(\epsilon|Z, A, X)\} \\
& = \EE\{(-1)^{1-A}q^*(Z,A,X)S(Y, W, Z,A,X)\EE(\epsilon|Z, A, X)\}.
\$
Thus, 
\$
& \frac{\partial \tau_t}{\partial t}\bigg|_{t = 0} = \EE\{[\{h^*(W,1,X) - h^*(W,0,X)\} \\
& ~~~~~~~~~~~~~~~~~~~~~~+ (-1)^{1-A}q^*(Z, A, X)\{\epsilon - \EE(\epsilon|Z, A, X)\}]\cdot S(Y,W,Z,A,X)\},
\$
which completes the proof.

\qed

\subsubsection{Proof of Theorem~\ref{thm:influence function}(ii)}
To establish that the influence function $\mathtt{IF}(\tau^*) = \mathtt{IF}(O;h^*, q^*, \ell^*, \tau^*)$ attains the semiparametric local efficiency bound, it suffices to show that $\mathtt{IF}(\tau^*)$ belongs to the tangent space of model $\cM_{\gamma}$ at the submodel $\cM_{{\rm eff}, \gamma}$.
As shown in the proof of Theorem~\ref{thm:influence function}(i), the model $\cM_\gamma$ imposes the restriction \eqref{h.derivative2}, which implies that the tangent space of $\cM_{{\rm eff}, \gamma}$ consists of the functions $S(O) \in L_2(O)$ such that there exists $g \in L_2(W,A,X)$ satisfying
\#
& \EE[d\{S(O) - S(Z,A,X)\}\epsilon^*] = \EE\{d(Z,A,X)g(W,A,X)\} \label{score.restriction}
\#
for all $d \in \HH_\gamma$, where $\epsilon^* = Y - h^*(W,A,X)$.
This reduces to finding a solution $g$ satisfying
\$
\Pi_\gamma(g) = \Pi_\gamma(\{S(O) - S(Z,A,X)\}\epsilon^*).
\$
At the submodel $\cM_{{\rm eff}, \gamma}$, Assumption~\ref{assump:surjectivity} implies that such a solution always exists for any $S(O) \in L_2(O)$.
Therefore, the influence function $\mathtt{IF}(\tau^*)$ belongs to the tangent space of $\cM_{ \gamma}$ at the submodel $\cM_{{\rm eff}, \gamma}$, which completes the proof.
\qed

\subsection{Proof of Theorem~\ref{thm:mr.estimator}} \label{appendix:mr.estimator}

We begin by imposing additional regularity assumptions of Appendix B of \cite{robins1994estimation} on the estimating equations for the nuisance parameters, which implies that $\hat \ba, \hat\bb, \hat\br$ and $\hat\bt$ have probability limits $\tilde \ba, \tilde \bb, \tilde \br$ and $\tilde \bt$, respectively, and are asymptotically normal.
Define $\zeta = (\ba, \bb, \bt, \br), \hat \zeta = (\hat \ba, \hat \bb, \hat \bt, \hat \br)$ and $\tilde\zeta = (\tilde\ba, \tilde\bb, \tilde\br, \tilde\bt)$.
We write 
\$
G(O;\zeta) & = G(O;\ba, \bb, \br, \bt)  = h(W, 1, X;\bb) - h(W,0,X;\bb) \\ 
& + d((-1)^{1-A}q(\bt);\ba)(Z,A,X)\{Y - h(W,A,X;\bb) - \ell(Z,X;\br)\}.
\$
Then, $\hat\tau_{{\rm fPMR}}$ is defined as $\hat\tau_{{\rm fPMR}} = \PP_n(G(O_i;\hat\ba, \hat\bb, \hat\br, \hat\bt))$.

We begin by showing that 
\#
\EE\{G(O;\tilde\ba, \tilde\bb, \tilde\br, \tilde\bt)\} = \tau^* \label{multiply.robust}
\#
under $\cM_{1, \gamma} \cup \cM_{2, \gamma} \cup \cM_{3, \gamma}$.
First, $\tilde\ba = \ba^*$ under $\cM_{1, \gamma}$, and thus, $\mathbf{c}_0(Z,A,X;\tilde\ba) \in \HH_\gamma$.
Therefore, for any value $\br$, we have
\$
\EE[\mathbf{c}_0(Z,A,X;\tilde\ba)\{Y - h(W,A,X;\bb^*) - \ell(Z, X;\br)\}] = 0,
\$
implying $\tilde\bb = \bb^*$.
Equation \eqref{multiply.robust} now follows because
\$
& \EE\{G(O;\tilde\ba, \tilde\bb, \tilde\br, \tilde\bt)\} = \EE\{G(O;\ba^*, \bb^*, \tilde\br, \tilde\bt)\} \\
& = \EE\{h(W,1,X;\bb^*) - h(W,0,X;\bb^*)\} \\
& ~~~~~~~~~~~~~~~~~~+ \EE[d((-1)^{1-A}q(\tilde\bt);\ba^*)(Z,A,X)\{Y - h(W,A,X;\bb^*) - \ell(Z,X;\tilde\br)\}] \\
& = \tau^* + \EE[d((-1)^{1-A}q(\tilde\bt);\ba^*)(Z,A,X)\{Y - h(W,A,X;\bb^*) - \ell(Z,X;\tilde\br)\}],
\$
where the last equation follows by Theorem~\ref{thm:outcome.bridge.identification}.
Note that $d(Z,A,X,q(\tilde\bt);\ba^*) \in \HH_\gamma$, so the second term in the last equation vanishes as $h(W,A,X;\bb^*)$ satisfies \eqref{outcome.bridge.moment} and $\ell \in \cS_\gamma$ for any $\br$.
Second, under $\cM_{2, \gamma}$,
\$
& \EE[\mathbf{c}_1(Z,A,X)\{Y - h(W,A,X;\bb^*) - \ell(Z,X;\br^*)\}] = 0, ~\mbox{ and }~ \\
& \EE[\mathbf{c}_0(Z,A,X;\ba)\{Y - h(W,A,X;\bb^*) - \ell(Z,X;\br^*)\}] = 0
\$
for any $\ba$, we have $\tilde\br = \br^*$ and $\tilde\bb = \bb^*$.
Since 
\$
\EE\{Y - h(W,A,X;\bb^*) - \ell(Z,X;\br^*)|Z,A,X\} = 0,
\$
we have
\$
& \EE\{G(O;\tilde\ba, \tilde\bb, \tilde\br, \tilde\bt)\} = \EE\{G(O;\tilde\ba, \bb^*, \br^*, \tilde\bt)\} \\
& = \tau^* + \EE[d((-1)^{1-A}q(\tilde\bt);\tilde\ba)(Z,A,X)\{Y - h(W,A,X;\bb^*) - \ell(Z,X;\br^*)\}] = \tau^*.
\$
Third, $\tilde\ba = \ba^*$ and $\tilde\bt = \bt^*$ as $(-1)^{1-A}q(Z,A,X;\bt^*) = (-1)^{1-A}q^*(Z,A,X) \in \HH_\gamma$ under $\cM_{3, \gamma}$.
By Theorem~\ref{thm:treatment.bridge.identification}, 
\$
\EE\{d((-1)^{1-A}q(\bt^*);\ba^*)(Z,A,X)Y\} = \EE\{(-1)^{1-A}q^* Y\} = \tau^*.
\$
Furthermore, 
\$
& \EE[d((-1)^{1-A}q(\bt^*);\ba^*)(Z,A,X)\{h(W,A,X;\tilde\bb) + \ell(Z,X;\tilde\br)\}] \\
& = \EE[(-1)^{1-A}q^*(Z,A,X)\{h(W,A,X;\tilde\bb) + \ell(Z,X;\tilde\br)\}] \\
& = h(W,1,X) - h(W,0,X),
\$
where the last equality follows by the assumption $\ell \in \cS_\gamma$ and $(-1)^{1-A}q^* \in \HH_\gamma$.
This implies \eqref{multiply.robust} under $\cM_{3,\gamma}$.

Now, assuming that $\partial G(O;\zeta)/\partial\zeta^\T$ satisfies the uniform law of large numbers in a neighborhood of $\tilde\zeta$, we apply a Taylor expansion along with the uniform law of large numbers, we obtain
\#
\sqrt{n}(\hat \tau_{{\rm fPMR}} - \tau^*) = \frac{1}{\sqrt{n}} \sn G(O_i;\tilde\zeta) + \EE\bigg\{\frac{\partial G(O;\zeta)}{\partial \zeta^\T}\bigg|_{\zeta = \tilde \zeta}\bigg\}\sqrt{n}(\hat \zeta - \tilde \zeta) + o_{P}(1), \label{dr.estimator.proof1}
\#
where $o_{P}(1)$ denotes a random variable converging to $0$ in probability.
From this expression, consistency and asymptotic normality follow directly, under the union model $\cM_{1, \gamma} \cup \cM_{2, \gamma} \cup \cM_{3, \gamma}$, by leveraging the multiply robust property $\EE\{G(O_i;\tilde\zeta)\} = 0$, the central limit theorem, and Slutsky's theorem.

To establish that $\hat \tau_{{\rm fPMR}}$ is semi-parametrically locally efficient in $\cM_{{\rm eff}, \gamma}$ at the intersection submodel $\cM_{1, \gamma} \cap \cM_{2, \gamma} \cap \cM_{3, \gamma}$, we remark that at the intersection submodel, $\tilde \zeta = \zeta^*$, where $\zeta^*$ denotes the true nuisance parameter.
Then, from the multiply robustness, it follows that $\EE\{\partial G(O;\zeta)/\partial \zeta^\T |_{\zeta = \zeta^*}\} = 0$.
Combining this with \eqref{dr.estimator.proof1}, we conclude that the influence function of $\hat \tau_{{\rm fPMR}}$ is given by  $\mathtt{IF}(O_i;h^*, q^*, \ell^*)$, which is semiparametric locally efficient in $\cM_{{\rm eff}, \gamma}$.
This completes the proof.

\qed

\subsection{Proof of Lemma~\ref{lem:change.of.measure}} \label{proof:change.of.measure}

We first claim the following:
Assuming that $(A, Z)$ are conditionally independent given $X$, $\HH_{\gamma}$ is characterized as
    \$
\HH_{\gamma} = \bigg\{d(Z,A,X) - \sum_{i = \gamma}^K \sum_{\myalpha \in \cP_{i}([K])} \alpha_i\EE\{d(Z,A,X)|Z_{-\myalpha}, X\} : d \in L_2(Z, A, X) \bigg\},
    \$
    where $\{\alpha_i\}_{i = \gamma}^K$ is inductively defined as 
    \#
\alpha_\gamma = 1 ~\mbox{ and }~ \sum_{j = \gamma}^i \binom{\gamma}{i - j} \alpha_j = 0 \mbox{ for } i > \gamma. \label{a.coefficients}
\#

To show this, we first need to show that this choice of $\alpha_i$ satisfies
\#
\EE\bigg[d(Z,A,X) - \sum_{i = \gamma}^K \sum_{\myalpha \in \cP_{i}([K])} \alpha_i\EE\{d(Z,A,X)|Z_{-\myalpha}, X\}\bigg|Z_{-\tilde\myalpha}, X\bigg] = 0 \label{intersection.proof1}
\#
for all $\tilde\myalpha \in \cP_{\gamma}([K])$.
Since $(Z, A)$ are conditionally independent given $X$, we remark that for any subset $\myalpha, \myalpha^\prime \subseteq [K]$, 
\#
\EE[\EE\{d(Z,A,X)|Z_{-\myalpha}, X\}|Z_{-\myalpha^\prime}, X] =  \EE\{d(Z,A,X)|Z_{-(\myalpha \cup \myalpha^\prime)}, X\}. \label{indep.cond.exp}
\#

Now, let $\tilde\myalpha \in \cP_{\gamma}([K])$ be fixed.
By \eqref{indep.cond.exp},
\#
& \EE\bigg[d(Z,A,X) - \sum_{i = \gamma}^K \sum_{\myalpha \in \cP_{i}([K])} \alpha_i\EE\{d(Z,A,X)|Z_{-\myalpha}, X\}\bigg|Z_{-\tilde\myalpha}, X\bigg] \nn \\
& = \EE\{d(Z,A,X)|Z_{-\tilde\myalpha}, X\} - \sum_{i = \gamma}^K \sum_{\myalpha \in \cP_{i}([K])} \alpha_i\EE\{d(Z,A,X)|Z_{-(\myalpha \cup \tilde\myalpha)}, X\} \nn \\
& ~~~~~~~~~~~~~~~~~~~~~~~~~~~~~~= \sum_{i =  \gamma}^K \sum_{\myalpha \in \cP_i([K]), \tilde\myalpha \subseteq \myalpha} b_{\myalpha}\EE\{d(Z,A,X)|Z_{-\myalpha}, X\} \label{intersection.proof2}
\#
for some coefficients $b_\myalpha$.
Hence, to establish \eqref{intersection.proof1}, we need to show that all the coefficients $b_\myalpha$ in \eqref{intersection.proof2} are zero.

For $i = \gamma$, the only $\myalpha \in \cP_\gamma([K])$ satisfying $\tilde\myalpha \subseteq \myalpha$ is $\tilde\myalpha$ itself.
By comparing the terms in \eqref{intersection.proof2}, we have
\$
\EE\{d(Z,A,X)|Z_{-\tilde\myalpha}, X\} - \alpha_\gamma\EE\{d(Z,A,X)|Z_{-\tilde\myalpha}, X\} = b_\gamma \EE\{d(Z, A, X)|Z_{-\tilde\myalpha}, X\},
\$
which implies $b_\gamma = 0$ as $\alpha_\gamma = 1$.
Now, consider $i > \gamma$ and fix $\myalpha \in \cP_i([K])$ with $\tilde\myalpha \subseteq \myalpha$.
To determine $b_\myalpha$ in \eqref{intersection.proof2}, we need to determine the subsets $\myalpha^\prime \subseteq [K]$ such that $\myalpha^\prime \cup \tilde\myalpha = \myalpha$ by \eqref{indep.cond.exp}.
For each $\gamma \leq j \leq i$, we count the number of subsets $\myalpha^\prime \in \cP_j([K])$ such that $\myalpha^\prime \cup \tilde\myalpha = \myalpha$ in the following.
Note that such $\myalpha^\prime$ should contain all elements in $\myalpha \setminus \tilde\myalpha$, where $A \setminus B$ is the set of elements in $A$ but not in $B$ for any subsets $A$ and $B$ in $[K]$. 
Since $|\myalpha\setminus\tilde\myalpha| = i - \gamma$, the remaining $j - (i - \gamma)$ elements of $\myalpha^\prime$ must be selected from $\tilde\myalpha$.
The number of such subsets $\myalpha^\prime$ is $\binom{\gamma}{j - i + \gamma} = \binom{\gamma}{i - j}$.
Thus, the coefficients of $\EE\{d|Z_{-\myalpha}, X\}$  with $|\myalpha| = i$ in the left-hand side of \eqref{intersection.proof2} is
\$
- \sum_{j = \gamma}^i \sum_{\myalpha^\prime \in \cP_i([K]), \myalpha^\prime \cup \tilde\myalpha = \myalpha} \alpha_i = - \sum_{j = \gamma}^i \binom{\gamma}{i - j} \alpha_i.
\$
By the choice of $\alpha_i$ in \eqref{a.coefficients}, this summation equals zero, which establishes \eqref{intersection.proof1}.
Thus, we have
\$
\HH_{\gamma} \supseteq \bigg\{d(Z,A,X) - \sum_{i = \gamma}^K \sum_{\myalpha \in \cP_{i}([K])} \alpha_i\EE\{d(Z,A,X)|Z_{-\myalpha}, X\} : d \in L_2(Z, A, X) \bigg\}.
\$

Conversely, for any $d \in \HH_{\gamma}$, we have
\$
d(Z,A,X) - \sum_{i = \gamma}^K \sum_{\myalpha \in \cP_{i}([K])} \alpha_i\EE\{d(Z,A,X)|Z_{-\myalpha}, X\} = d(Z, A, X).
\$
Thus, 
\$
\HH_{\gamma} \subseteq \bigg\{d(Z,A,X) - \sum_{i = \gamma}^K \sum_{\myalpha \in \cP_{i}([K])} \alpha_i\EE\{d(Z,A,X)|Z_{-\myalpha}, X\} : d \in L_2(Z, A, X) \bigg\},
\$
which completes the proof.

Now, we consider the general setting where $(Z,A)$ may not be conditionally independent given $X$.
Let $d \in L_2(Z, A, X)$ be a function such that $v(Z,A,X) \in L_2(Z,A,X)$, where $v$ is defined as follows:
\$
v(Z, A, X) := \frac{f^*}{f}\bigg\{d - \sum_{i = \gamma}^K \sum_{\myalpha \in \cP_{i}([K])} \alpha_i\EE^*(d|Z_{-\myalpha}, X) \bigg\}. 
\$
Since $f^*$ is absolutely continuous with respect to $f$, $f^*/f$ is the Radon-Nikod{\' y}m derivative of $f^*$ with respect to $f$.
Thus, for any $\tilde \myalpha \in \cP_\gamma([K])$, we have
\$
\EE(v|Z_{-\myalpha}, X) & = \EE\bigg[ \frac{f^*}{f}\bigg\{d - \sum_{i = \gamma}^K \sum_{\myalpha \in \cP_{i}([K])} \alpha_i\EE^*(d|Z_{-\myalpha}, X) \bigg| Z_{-\tilde\myalpha}, X\bigg] \\
& = \frac{1}{f(Z_{-\tilde\myalpha}, X)}\EE^*\bigg\{ d - \sum_{i = \gamma}^K \sum_{\myalpha \in \cP_{i}([K])} \alpha_i\EE^*(d|Z_{-\myalpha}, X)\bigg| Z_{-\tilde\myalpha}, X \bigg\} = 0.
\$
This implies that $v \in \HH_{\gamma}$, that is,
\$
\HH_{\gamma}  \supseteq \bigg\{\frac{f^*}{f}\bigg\{d - \sum_{i = \gamma}^K \sum_{\myalpha \in \cP_{i}([K])} \alpha_i\EE^*(d|Z_{-\myalpha}, X) \bigg\}: d \in L_2(Z, A, X)\bigg\} \cap L_2(Z, A, X).
\$

To prove the reverse inclusion, let $v(Z,A,X)$ be any function in $\HH_{\gamma}$.
As we assume that $f/f^* \in L_2(Z,A,X)$, the function 
\$
d(Z, A, X) := \frac{f}{f^*}\cdot v(Z, A, X)
\$
is well-defined and belongs to $L_2(Z,A,X)$.
Then, for any $\myalpha \in \cP_\gamma([K])$, we have 
\$
\EE^*(d(Z, A, X)|Z_{-\myalpha}, X) = \frac{f(Z_{-\myalpha}, X)}{f^*(Z_{-\myalpha}, X)}\EE\{v(Z, A, X)|Z_{-\myalpha},  X\} = 0,
\$
where the last equation follows as $v \in \HH_{\gamma}$.
Thus, 
\$
& \frac{f^*}{f}\bigg\{d - \sum_{i = \gamma}^K \sum_{\myalpha \in \cP_{i}([K])} \alpha_i\EE^*(d|Z_{-\myalpha}, X) \bigg\} = \frac{f^*}{f} d(Z, A, X) = v(Z,A,X).
\$
This completes the proof.

\qed

\bibliographystyle{plainnat}
\bibliography{sample}

\end{document}